\documentclass[runningheads, envcountsame, a4paper]{llncs}
\usepackage{comment}
\usepackage{cite}

\usepackage{etoolbox}
\usepackage{parskip}
\def\@testdef #1#2#3{
  \def\reserved@a{#3}\expandafter \ifx \csname #1@#2\endcsname
  \reserved@a  \else
 \typeout{^^Jlabel #2 changed:^^J%
\meaning\reserved@a^^J
\expandafter\meaning\csname #1@#2\endcsname^^J}%
 \@tempswatrue \fi}
\usepackage[T1]{fontenc}
\usepackage{amssymb}
\usepackage{amsmath}
\usepackage{quantikz}
\usepackage{hyperref}
\usepackage{tikz}
\usetikzlibrary{calc}
\usetikzlibrary{arrows.meta}
\usepackage{subfigure}
\usepackage{textcomp}
\usepackage[capitalize]{cleveref}
\crefname{proof}{proof}{proofs}
\Crefname{proof}{Proof}{Proofs} 
\usepackage{paralist}
\usepackage{xcolor}
\usepackage{booktabs}
\usepackage{wrapfig}
\usepackage{marvosym}
\usepackage{graphicx}
\usepackage{empheq}
\usepackage{proof}
\usepackage{listings}
\usepackage{xcolor}
\usepackage{makecell}
\usepackage{tabularx}
\usepackage{soul}

\usepackage{braket}
\let\set=\sett

\usepackage{syntax}

\usetikzlibrary{arrows}
\usetikzlibrary{positioning}
\usetikzlibrary{graphs}
\usetikzlibrary{backgrounds}
\usepackage{forest}
\usepackage[linesnumbered,noend,ruled]{algorithm2e}

\usepackage{newtxtext}
\usepackage{newtxmath}
\newcommand{\showcomment}[1]{#1}
\newcommand{\ol}[1]{\showcomment{\textcolor{orange}{\ifmmode \text{[OL: #1]}\else [OL: #1] \fi}}}
\newcommand{\yfc}[1]{\showcomment{\textcolor{purple}{\ifmmode \text{[YFC: #1]}\else [YFC: #1] \fi}}}
\newcommand{\km}[1]{\showcomment{\textcolor{green}{\ifmmode \text{[KM: #1]}\else [KM: #1] \fi}}}
\newcommand{\wl}[1]{\showcomment{\textcolor{blue}{\ifmmode \text{[#1]}\else [#1] \fi}}}
\newcommand{\ja}[1]{\showcomment{\textcolor{red}{\ifmmode \text{[JA: #1]}\else [JA: #1] \fi}}}

\newcommand{\topof}[1]{\mathtt{top}(#1)}
\newcommand{\symof}[1]{\mathtt{amp}(#1)}
\newcommand{\chof}[1]{\mathtt{ch}(#1)}

\newcommand{\leftof}[1]{\mathtt{left}(#1)}
\newcommand{\rightof}[1]{\mathtt{right}(#1)}

\tikzset{st/.style={font=\ttfamily,shape=rectangle,rounded corners=.5em,draw=black,fill=gray!30,inner xsep=.2em,inner ysep=0em,text height=2.2ex,text depth=1.0ex}}
\tikzset{leaf/.style={font=\ttfamily,shape=rectangle,rounded corners=.5em,draw=black,fill=blue!20,inner xsep=.2em,inner ysep=0em,text height=2.2ex,text depth=1.0ex}}

\newcommand{\hide}[1]{}

\newcommand{\autoqq}[0]{\textsc{AutoQ~2.0}\xspace}

\newcommand{\autoqqq}[0]{\textsc{AutoQ}\xspace}

\newcommand{\Arity}[0]{\ifmmode \mathbf{Arity} \else \textbf{Arity} \fi}

\newcommand{\height}[0]{\mathrm{ht}}

\lstdefinelanguage{OpenQASM}{
  keywords=[1]{qubit, bit, gate, barrier, measure, reset},
  keywords=[2]{cx, ccx, x, h, cu},
  sensitive=true,
  morecomment=[l]{//},
  morestring=[b]',
}

\lstdefinestyle{OpenQASMStyle}{
  language=OpenQASM,
  basicstyle=\small\ttfamily,
  keywordstyle=[1]\color{blue},
  keywordstyle=[2]\color{green!60!black},
  commentstyle=\color{gray}\itshape,
  stringstyle=\color{red},
  numbers=left,
  numberstyle=\color{gray},
  stepnumber=1,
  numbersep=15pt,
  backgroundcolor=\color{gray!10},
  frame=single,
  framesep=2mm,
  rulecolor=\color{black},
  breaklines=true,
  showstringspaces=false,
}

\newcommand{\bluelab}[1]{\text{\small {\protect\tikz[baseline,node distance=2mm]{%
  \protect\node[shape=rectangle,rounded corners=.5em,draw=black,inner xsep=.3em,inner ysep=0em,text height=2.3ex,text depth=1.0ex,fill=blue!20] (nd) at (0,.64ex) {\hspace{.1em}\texttt{\upshape\strut$#1$}\hspace{.1em}\strut};%
  }}}}

\newcommand{\bool}[0]{\mathbb{B}}

\newcommand{\figBVcircuit}[0]{
\begin{wrapfigure}[12]{r}{0.4\textwidth}
    \vspace{-10mm}
    \centering
    \scalebox{0.7}{
    \begin{quantikz}
    \lstick{$\ket{s_1}$}\gategroup[wires=3,steps=7,style={rounded corners,fill=blue!10,draw opacity=0},background]{} && \ctrl{6} \gategroup[wires=7,steps=3,style={dashed,
    rounded corners,fill=blue!10,draw opacity=0},background]{} &&&& \\
    \lstick{$\ket{s_2}$}\gategroup[wires=1,steps=7,style={rounded corners,fill=blue!10,draw opacity=0},background]{} &&& \ctrl{5} &&&\\
    \lstick{$\ket{s_3}$}\gategroup[wires=1,steps=7,style={rounded corners,fill=blue!10,draw opacity=0},background]{} &&&& \ctrl{4} &&\\
    \lstick{$\ket{0}$} & \gate{H} & \ctrl{0} &&& \gate{H} & \\
    \lstick{$\ket{0}$} & \gate{H} && \ctrl{0} && \gate{H}  & \\
    \lstick{$\ket{0}$} & \gate{H} &&& \ctrl{0} &\gate{H}  &\\
    \lstick{$\ket{1}$} & \gate{H} & \targ{} & \targ{}  & \targ{} & \gate{H} &
    \end{quantikz}
    }
    \vspace{-2mm}
    \caption{3-qubit BV circuit.}
    \label{fig:bv-circuit}    
\end{wrapfigure}
}

\newcommand{\figMCToffoli}[0]{
\begin{wrapfigure}[12]{r}{0.45\textwidth}
 \vspace{-8mm}
\resizebox*{0.43\textwidth}{!}{
\begin{quantikz}
\lstick{$\ket{c_1}$} & \ctrl{5} & \qw & \qw & \qw & \qw & \qw & \qw & \qw & \ctrl{5} & \qw\\
\lstick{$\ket{c_2}$} & \ctrl{1} & \qw & \qw & \qw & \qw & \qw & \qw & \qw & \ctrl{1} & \qw\\
\lstick{$\ket{c_3}$} & \qw & \ctrl{3} & \qw & \qw & \qw & \qw & \qw & \ctrl{3} & \qw & \qw\\
\lstick{$\ket{c_4}$} & \qw & \qw & \ctrl{4} & \qw & \qw & \qw & \ctrl{4} & \qw & \qw & \qw\\
\lstick{$\ket{c_5}$} & \qw & \qw & \qw & \ctrl{3} & \qw & \ctrl{3} & \qw & \qw & \qw & \qw\\
\lstick{$\ket{0}$} & \targ{} & \ctrl{1} & \qw & \qw & \qw & \qw & \qw & \ctrl{1} & \targ{} & \qw\\
\lstick{$\ket{0}$} & \qw & \targ{} & \ctrl{1} & \qw & \qw & \qw & \ctrl{1} & \targ{} & \qw & \qw\\
\lstick{$\ket{0}$} & \qw & \qw & \targ{} & \ctrl{1} & \qw & \ctrl{1} & \targ{} & \qw & \qw & \qw\\
\lstick{$\ket{0}$} & \qw & \qw & \qw & \targ{} & \ctrl{1} & \targ{} & \qw & \qw & \qw & \qw\\
\lstick{$\ket{t}$} & \qw & \qw & \qw & \qw & \targ{} & \qw & \qw & \qw & \qw & \qw
\end{quantikz}
}
\caption{Multi-controlled Toffoli with 5 control qubits. Standard Toffoli gates have two $\bullet$ as controls and one $\oplus$ as the target.}
\label{fig:cccx}
\end{wrapfigure}
}

\usepackage{multicol}
\usepackage{multirow}
\usepackage{colortbl}
\usepackage{xcolor}

\usepackage{thmtools, thm-restate}
\InputIfFileExists{cutmargins.tex}{%
	\pdfpagesattr{/CropBox [125 82 490 688]} % LNCS page made even biger
}{%
}

\Crefname{algocf}{Algorithm}{Algorithms}

\def\orcidID#1{\smash{\href{http://orcid.org/#1}{\protect\raisebox{-1.25pt}{\protect\includegraphics[alt={ORCID logo}]{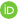}}}}}

\title{A Practical Specification Language for Automatic Quantum Program Verification (Technical Report)}

\author{Wei-Lun Tsai\inst{1,2}\orcidID{0009-0003-5832-0867}\Letter \and
Yu-Fang Chen\inst{1}\orcidID{0000-0003-2872-0336}\Letter \and
Ond\v{r}ej Leng\'{a}l\inst{3}\orcidID{0000-0002-3038-5875}\Letter}
\authorrunning{W.-L. Tsai, Y.-F. Chen, and O. Leng\'{a}l}
\institute{Institute of Information Science, Academia Sinica, Taipei, Taiwan \and
Graduate Institute of Electronics Engineering, National Taiwan University, Taipei, Taiwan \and
Faculty of Information Technology, Brno University of Technology, Brno, Czech Republic\\
\email{alan23273850@gmail.com, gulu0724@gmail.com, lengal@fit.vutbr.cz}}

\begin{document}

\maketitle
% \vspace{-10mm}
\begin{abstract}
Hoare-style verification provides a principled foundation for reasoning about the correctness of quantum programs, but existing approaches do not allow fully automatic verification. While automata-based verification scales well when specifications are given directly as automata, prior frameworks incur exponential blow-up when translating high-level set-based assertions into automata, which severely limits practicality. We introduce an extended set-based specification language and a specification-to-automata translation algorithm whose complexity is linear in the number of qubits, enabled by controlled automaton construction and qubit reordering. The resulting compact automata enable fully automatic Hoare-style verification of fixed-qubit quantum programs at previously infeasible scales, while substantially improving expressiveness without compromising efficiency.
\end{abstract}

%%%%%%%%%%%%%%%%%%%%%%%%%%%%%%%%%%%%%%%%%%%%%%
\section{Introduction}\label{sec:introduction}

Recent advances in quantum hardware, exemplified by early demonstrations targeting quantum supremacy \cite{arute2019quantum} and evidence of quantum utility on processors exceeding 100 qubits \cite{DBLP:journals/nature/KimEAWBRNWZTK23}, have accelerated the development of quantum programming. This progress is shifting the field from predominantly theoretical studies toward practical applications in cryptography \cite{DBLP:conf/focs/Shor94}, finance \cite{DBLP:journals/quantum/StamatopoulosES20}, and optimization \cite{harrigan2021quantum}. As quantum programs grow in size and complexity, ensuring their correctness becomes increasingly critical.

Despite significant progress in quantum program logics and verification frameworks in the last decade~\cite{DBLP:journals/toplas/Ying11,DBLP:conf/pldi/ZhouYY19,Coecke2011,DBLP:conf/esop/CharetonBBPV21,DBLP:journals/corr/abs-1805-06908,pldi23,DBLP:conf/cav/ChenCLLT23,DBLP:journals/pacmpl/AbdullaCCHLLLT25,DBLP:conf/tacas/ChenCHHLLT25,cacm25,DBLP:conf/cade/ChenRT23,iccad24}, a gap remains between theory and practice. In particular, most existing approaches do not provide a practical specification language that allows users to state nontrivial correctness properties and have them verified fully automatically, without interactive proof construction. Bridging this gap is essential for making formal verification a usable component of the quantum software development workflow.
\emph{Quantum Hoare logic} (QHL) is a widely used framework for reasoning about the correctness of quantum programs. In QHL, program behavior is specified using triples $\{P\}\,C\,\{Q\}$, where the precondition $P$ and postcondition $Q$ are represented as \emph{Hermitian operators} over the underlying Hilbert space \cite{DBLP:journals/toplas/Ying11}. In the past, a variety of quantum Hoare-style systems have been developed~\cite{DBLP:journals/toplas/Ying11,DBLP:conf/pldi/ZhouYY19,DBLP:journals/tqc/FengY21,DBLP:conf/lics/Unruh19,sundaram2025hoaremeetsheisenberglightweight,DBLP:conf/ecoop/0002ZCNLC024,yu2025logicapproximatequantitativereasoning,10.1145/3770083,DBLP:journals/corr/abs-2109-02198,DBLP:conf/cav/YanJY24,DBLP:journals/pacmpl/YanJY22,DBLP:conf/qsw/LewisZS24}, providing sound semantic foundations for formal reasoning about program correctness.

However, existing quantum Hoare-style approaches do not yet support practical push-button verification. Writing assertions that capture nontrivial correctness properties typically leads to complex proof obligations that require substantial manual effort in interactive theorem provers. For example, even for Grover's search algorithm~\cite{Grover96}, formal verification for an arbitrary number of qubits has only been achieved through large, hand-crafted proof developments in systems such as Isabelle/HOL and Coq, comprising hundreds to thousands of lines of proof scripts \cite{DBLP:conf/cav/LiuZWYLLYZ19,DBLP:journals/pacmpl/ZhouBSLY23}.
The limitations of existing quantum Hoare-style verification frameworks suggest that the choice of assertion representation plays a decisive role in enabling automation. In this work, we advocate assertions based on \emph{sets of quantum states} as a powerful and natural alternative.

Set-based assertions provide a natural and expressive way to specify correctness properties of quantum programs. Unlike traditional assertion representations, they can directly describe families of quantum states that satisfy explicit quantitative constraints. For example, the postcondition of Grover's algorithm can be specified as the set
$
\bigcup_{|\alpha|^2\ >\ 0.8}\Big\{ \alpha\ket w + \beta\sum_{\substack{i\neq w}}\ket i \Big\}$
for some $n \ge 2$~\cite{younes2008strengthweaknessgroversquantum} and marked item $w\in\set{0,1}^n$. This specification directly captures the intended probabilistic guarantee without resorting to indirect encodings or symbolic reasoning about amplitudes.
Such specifications are particularly well-suited to automated verification: correctness properties are stated as concrete constraints over sets of states, which in principle can be manipulated algorithmically. Building on this idea, \textsc{AutoQ}~\cite{DBLP:conf/cav/ChenCLLT23} represents set-based assertions using finite automata and SMT constraints, and performs Hoare-style verification through algorithmic transformations over these representations. Once the specification is given, the verification process proceeds without interactive theorem proving.

However, the practical applicability of this approach is limited by the cost of translating specifications into automata. In prior work, the size of the generated automata can grow exponentially with respect to the number of qubits. This exponential blow-up severely restricts verification to very small instances, despite the conceptual suitability of set-based assertions for automation.
In this work, we address this scalability bottleneck by introducing a new specification-to-automata translation algorithm whose complexity is linear in the qubit number. Before building automata, it reorganizes the specification into a tensor product of smaller, mostly independent components (first at the variable level, then at the qubit level), and only then constructs and composes the corresponding automata. As a result, extended set-based specifications can be compiled into compact automata even for nontrivial cases, enabling fully automatic verification at a scale that was previously infeasible.

We implemented the new translation algorithm and compared it against the method in~\cite{DBLP:conf/cav/ChenCLLT23}. The results show that our algorithm substantially outperforms the prior approach. For instance, we translate the functional-correctness specification of a 32-qubit Grover circuit into an automaton in under one second, whereas~\cite{DBLP:conf/cav/ChenCLLT23} does not complete the translation within five minutes.

\paragraph{Related Work.}
Within quantum Hoare logic, two principal representations of predicates have emerged. The first represents assertions as Hermitian operators, following the work of D'Hondt and Panangaden \cite{DBLP:journals/mscs/DHondtP06} and further developed by Ying \cite{DBLP:journals/toplas/Ying11}. This formulation enables quantitative reasoning over mixed states, supporting properties such as success probabilities and expected values. The second represents assertions as projections, or closed subspaces of Hilbert spaces, rooted in the quantum logic of Birkhoff and von Neumann \cite{BirkhoffNeumann1936} and later applied to quantum Hoare logic for qualitative reasoning \cite{DBLP:conf/pldi/ZhouYY19}.
While both representations are mathematically well-founded, they pose challenges for automation. The implementation is often via interactive theorem provers, and the proof requires significant manual efforts~\cite{DBLP:conf/cav/LiuZWYLLYZ19,DBLP:journals/pacmpl/ZhouBSLY23}.

Set-based assertions have recently been explored as an alternative specification mechanism, aiming to better support automation by treating correctness properties as explicit sets of quantum states. Prior work has shown that such assertions can be verified algorithmically when given directly as automata representations \cite{DBLP:conf/cav/ChenCLLT23,pldi23,cacm25,DBLP:journals/pacmpl/AbdullaCCHLLLT25,popl26}. However, existing approaches suffer from severe scalability issues when translating high-level specifications into automata, often incurring exponential blow-up.

\section{Background}\label{sec:Background}
%In this section, we provide the background necessary to understand this work.

\subsection{Quantum Computing}\label{sec:Quantum States in Trees}
In quantum computing, an $n$-qubit \emph{quantum state} is a superposition of all $2^n$ computational basis states:
$\ket{\psi} = \sum_{i\in\{0,1\}^n} c_i \ket{i}$,
where each complex \emph{amplitude} $c_i$ is associated with the basis state $\ket{i}$ and satisfies the normalization condition
$\sum_i |c_i|^2 = 1$.
For example,
$\tfrac12\ket{00} - \tfrac12\ket{01} + \tfrac{i}{\sqrt{2}}\ket{10}$
is a valid two-qubit state.

A quantum state can be viewed structurally as a \emph{perfect binary tree} of height~$n$.
The $k$-th level corresponds to the $k$-th qubit, and each basis string in $\{0,1\}^n$ determines a unique root-to-leaf path, taking the left branch for $0$ and the right branch for $1$.
The leaf reached by this path stores the amplitude of the corresponding basis state.
See \Cref{fig:stateA,fig:stateB} for illustrations.

To compose quantum systems, let
$\ket{\psi} = \sum_{s\in\{0,1\}^n} a_s\ket{s}$ be an $n$-qubit state and
$\ket{\phi} = \sum_{t\in\{0,1\}^m} b_t\ket{t}$ an $m$-qubit state.
Their tensor product is the $(n+m)$-qubit state
$\ket{\psi}\otimes\ket{\phi}
= \sum_{s\in\{0,1\}^n,t\in\{0,1\}^m} a_s b_t \ket{st}$,
where $\ket{st}$ denotes concatenation of basis strings.
In the tree representation, this corresponds to replacing each leaf labeled $a_s$ in the tree of $\ket{\psi}$ by a copy of the tree of $\ket{\phi}$, scaled by $a_s$.
We extend this operation elementwise to sets of quantum states.
For sets $S_1$ and $S_2$, define
$S_1 \otimes S_2
= \{\ket{\psi}\otimes\ket{\phi} \mid \ket{\psi}\in S_1,\ \ket{\phi}\in S_2\}.$

Finally, quantum computation proceeds by applying \emph{quantum gates}, which are unitary operators mapping quantum states to quantum states while preserving normalization.
Quantum gates are the fundamental building blocks of \emph{quantum circuits}.

\subsection{Level-Synchronized Tree Automata}\label{sec:Level-Synchronized Tree Automata}

We must choose an automata model as the target formalism for specification translation.
We adopt \emph{Level-Synchronized Tree Automata} (LSTA)~\cite{DBLP:journals/pacmpl/AbdullaCCHLLLT25}.
Compared with standard tree automata~\cite{cacm25,pldi23}, LSTAs provide a more compact encoding of quantum states and are directly supported by the verification tool AutoQ~\cite{DBLP:conf/cav/ChenCLLT23,DBLP:conf/tacas/ChenCHHLLT25}.
By translating specifications into LSTAs, we can directly reuse AutoQ's decision procedures and obtain a fully automatic, end-to-end verification workflow.
Thus, LSTAs are not only expressive and succinct, but also practically well suited for tool-supported verification.
As discussed in \Cref{sec:Quantum States in Trees}, a quantum state can be viewed structurally as a perfect binary tree.
LSTAs compactly represent \emph{sets} of such trees by sharing common substructures.
This section introduces the formal definition of LSTAs, their semantics, and supported operations.

\subsubsection{Definition.}
Let $\mathbb K$ be a commutative nonunital semiring (i.e., a structure closed under addition and multiplication, equipped with an additive identity $0_\mathbb K$ that is absorbing for multiplication). A \emph{level-synchronized tree automaton (LSTA)}~\cite{DBLP:journals/pacmpl/AbdullaCCHLLLT25} over $\mathbb K$ is a tuple $\mathcal{A} = \langle Q, V, \Delta, r\rangle_\mathbb K$ where \(Q\) is a set of \emph{states}, \(V\) is a set of \emph{variables}, \(\Delta\) is a set of \emph{transitions}, and \(r \in Q\) is the \emph{root state} (or \emph{starting state}). We often omit the subscript $\mathbb K$ when the semiring is clear from the context. The transition set \(\Delta\) is divided into two disjoint nonempty subsets: \(\Delta_\textit{in}\) (\emph{internal transitions}) and \(\Delta_\textit{ex}\) (\emph{external or leaf transitions}).

An internal transition \(\delta \in \Delta_\textit{in}\) has the form $q \xrightarrow{C} (q_1, q_2),$ and a leaf transition \(\delta \in \Delta_\textit{ex}\) has the form $q \xrightarrow{C} e,$ where \(q, q_1, q_2 \in Q\), \(e \in \mathbb K[V]\) (i.e., a polynomial over \(V\) with coefficients in $\mathbb K$), and \(C \subseteq \mathbb{N}\) is a nonempty finite set of \emph{choices}. We refer to \(q\), \(q_1\), \(q_2\), \(e\), and \(C\) as the \emph{top state}, \emph{left child}, \emph{right child}, \emph{amplitude}, and \emph{choices} of the transition~\(\delta\), and denote them as \(\topof{\delta}\), \(\leftof{\delta}\), \(\rightof{\delta}\), \(\symof{\delta}\), and \(\chof{\delta}\), respectively. 

We define \emph{root transitions} as $\Delta_r = \{ \delta \in \Delta_\textit{in} \mid \topof{\delta} = r \}$. To ensure deterministic resolution of transitions, LSTAs satisfy
\emph{choice disjointness}:
for any two distinct transitions $\delta_1,\delta_2$ with the same top state,
their choice sets are disjoint.
The size of an LSTA $\mathcal{A}$, denoted $|\mathcal{A}|$, is defined as $|\Delta|$.

\subsubsection{Sets of Quantum States.}
In this work, an LSTA is used as an internal representation of a set of quantum states.
Intuitively, an LSTA $\mathcal{A}$ encodes a set $\mathcal{L}(\mathcal{A})$ of perfect binary trees, each corresponding to a
quantum state as described in \Cref{sec:Quantum States in Trees}.
A quantum state
$\ket{\psi}=\sum_{s\in\{0,1\}^n} a_s\ket{s}$
belongs to $\mathcal{L}(\mathcal{A})$ if and only if there exists a sequence of choices
$c_1,\dots,c_n,c_0\in\mathbb{N}$ that induces a perfect binary tree whose leaf amplitudes
match the coefficients $a_s$.
Such a tree, if it exists, is unique by choice disjointness.

Formally, the sequence of choices must satisfy the following conditions.
\begin{enumerate}[(1)]
\item
The choices $c_1,\ldots,c_n,c_0$ induce a root-to-leaf path for every branch of the tree. For each basis string $s=b_1\ldots b_n\in\{0,1\}^n$, there exists a unique sequence of states
$q_0,q_1,\ldots,q_n$ and a value $v_s$ such that $q_0=r$ and:
\begin{inparaenum}
    \item for each internal level $1\le i\le n$, there exists a unique internal transition
    $\delta\in\Delta_{\textit{in}}$ with
    $\topof{\delta}=q_{i-1} \;\land\;c_i\in\chof{\delta}\;\land\; \mathit{ite}(b_i=0, \leftof{\delta}=q_i, \rightof{\delta}=q_i\bigr)$;
    \item at the leaf level, there exists a unique external transition
    $\delta\in\Delta_{\textit{ex}}$ with
    $\topof{\delta}=q_n
        \;\land\;
        c_0\in\chof{\delta}
        \;\land\;
        \symof{\delta}=v_s$.
    
\end{inparaenum}
A transition $\delta$ is said to be \emph{enabled} by a choice $c$ if $c\in\chof{\delta}$.
\item For all $s\in\{0,1\}^n$, the induced value satisfies $v_s=a_s$.
\end{enumerate}

\begin{comment}
Equivalently, this sequence must satisfy the following conditions:
\begin{enumerate}[(1)]
    \item \textbf{Structural Validity:} The choices $c_1, \ldots, c_n, c_0$ ensure the existence of a valid root-to-leaf path for every branch of the tree. Specifically, for every $s = b_1 \ldots b_n \in \{0,1\}^n$, there must exist a unique sequence of states $q_0, q_1, \ldots, q_n$ followed by an amplitude $v_s$ such that $q_0 = r$; for each internal level $1 \le i \le n$, there exists a unique internal transition $\delta \in \Delta_\textit{in}$ satisfying:
    \[
        \topof{\delta} = q_{i-1} \land c_i \in \chof{\delta} \land (\text{if } b_i=0 \text{ then } \leftof{\delta}=q_i \text{ else } \rightof{\delta}=q_i);
    \]
    and for the leaf level, there exists a unique external transition $\delta \in \Delta_\textit{ex}$ satisfying:
    \[
        \topof{\delta} = q_n \land c_0 \in \chof{\delta} \land \symof{\delta} = v_s. % \vert_\nu
    \]
    % (Here, $e\vert_\nu$ denotes the evaluation of the symbolic expression $e$ under the valuation $\nu$; and we say that a transition $\delta$ is \emph{enabled} by choice $c$ if $c \in \chof{\delta}$.)
    (In this context, we say that a transition $\delta$ is \emph{enabled} by choice $c$ if $c \in \chof{\delta}$.)
    \item \textbf{Amplitude Matching:} $v_s = a_s$ for all $s$.
\end{enumerate}
\end{comment}

\begin{example}
Consider a set of quantum states $\{\frac1{\sqrt2}\ket{00} - \frac1{\sqrt2}\ket{01}, \frac{i}{\sqrt2}\ket{10} - \frac{i}{\sqrt2}\ket{11}\}$ and an LSTA $\langle Q, \emptyset, \Delta_{\textit{in}}\cup\Delta_{\textit{ex}}, r\rangle,$ where \(Q = \{q_i \mid 0 \le i\le 8\}\), \(\Delta_{\textit{in}} = \{q_0 \xrightarrow{\{1\}} (q_1, q_2),\ q_0 \xrightarrow{\{2\}} (q_2, q_3),\ q_1 \xrightarrow{\{1\}} (q_4, q_5),\ q_2 \xrightarrow{\{1\}} (q_6, q_6),\ q_3 \xrightarrow{\{1\}} (q_7, q_8)\}, \Delta_{\textit{ex}} = \{\ q_4 \xrightarrow{\{1\}} \frac1{\sqrt2},\ q_5 \xrightarrow{\{1\}} \frac{-1}{\sqrt2},\ q_6 \xrightarrow{\{1\}} 0,\ q_7 \xrightarrow{\{1\}} \frac{i}{\sqrt2},\ q_8 \xrightarrow{\{1\}} \frac{-i}{\sqrt2}\}\), \(r = q_0\). In \Cref{fig:Representing a set of quantum states with an LSTA}, we depict the two quantum states as perfect binary trees and present their representative LSTA vertically for a better visualization of how it is used to induce the trees. All transitions whose choices include $1$ collectively form the tree in \Cref{fig:stateA}, whereas replacing the root transition with the alternative one yields the tree in \Cref{fig:stateB}. These are the only trees that can be induced by \Cref{fig:LSTA}.

%alt={A composite figure with three subfigures detailing the tree and automaton representations of two quantum states. Subfigure (a) shows a perfect binary tree for the state 1/sqrt(2)|00> minus 1/sqrt(2)|01>. The root node q0 has a left child q1 and a right child q2. Node q1 has a left child q4 (with a leaf amplitude of 1/sqrt(2)) and a right child q5 (with a leaf amplitude of -1/sqrt(2)). Node q2 has two children both named q6, each with a leaf amplitude of 0. Subfigure (b) shows a perfect binary tree for the state i/sqrt(2)|10> minus i/sqrt(2)|11>. The root node q0 has a left child q2 and a right child q3. Node q2 has two children both named q6, each with a leaf amplitude of 0. Node q3 has a left child q7 (with a leaf amplitude of i/sqrt(2)) and a right child q8 (with a leaf amplitude of -i/sqrt(2)). Subfigure (c) shows the LSTA representation combining these structures. From the root node q0, a transition with choice {1} leads to the pair (q1, q2), and a transition with choice {2} leads to the pair (q2, q3). At the next level, all transitions use choice {1}: q1 leads to (q4, q5), q2 leads to (q6, q6), and q3 leads to (q7, q8). At the leaf level, transitions with choice {1} map the nodes to their respective amplitudes: q4 maps to 1/sqrt(2), q5 maps to -1/sqrt(2), q6 maps to 0, q7 maps to i/sqrt(2), and q8 maps to -i/sqrt(2).}
\begin{figure}[tb]
\resizebox{\textwidth}{!}{
\subfigure[$\frac1{\sqrt2}\ket{00} - \frac1{\sqrt2}\ket{01}$]{
    \begin{tikzpicture}[anchor=base]
    \node {$q_0$} [sibling distance = 1.4cm, level distance = 0.8cm]
        child {node {$q_1$} edge from parent [solid, sibling distance = .7cm]
            child {node {$q_4$} edge from parent [solid]
                child {node {$\frac{1}{\sqrt2}$} edge from parent [solid]}
            }
            child {node {$q_5$} edge from parent [solid]
                child {node {$\frac{-1}{\sqrt2}$} edge from parent [solid]}
            }
        }
        child {node {$q_2$} edge from parent [solid, sibling distance = .7cm]
            child {node {$q_6$} edge from parent [solid]
                child {node {$0$} edge from parent [solid]}
            }
            child {node {$q_6$} edge from parent [solid]
                child {node {$0$} edge from parent [solid]}
            }
        };
    \end{tikzpicture}
    \label{fig:stateA}
}
\subfigure[$\frac{i}{\sqrt2}\ket{10} - \frac{i}{\sqrt2}\ket{11}$]{
    \begin{tikzpicture}[anchor=base]
    \node {$q_0$} [sibling distance = 1.4cm, level distance = 0.8cm]
        child {node {$q_2$} edge from parent [solid, sibling distance = .7cm]
            child {node {$q_6$} edge from parent [solid]
                child {node {$0$} edge from parent [solid]}
            }
            child {node {$q_6$} edge from parent [solid]
                child {node {$0$} edge from parent [solid]}
            }
        }
        child {node {$q_3$} edge from parent [solid, sibling distance = .7cm]
            child {node {$q_7$} edge from parent [solid]
                child {node {$\frac{i}{\sqrt2}$} edge from parent [solid]}
            }
            child {node {$q_8$} edge from parent [solid]
                child {node {$\frac{-i}{\sqrt2}$} edge from parent [solid]}
            }
        };
    \end{tikzpicture}
    \label{fig:stateB}
}
\subfigure[The LSTA representation]{
\resizebox{!}{0.55\height}{
    \begin{tikzpicture}[>=stealth',node distance=20mm,rotate=37]
  \pgfsetlinewidth{1bp}
  \tikzstyle{bddnode}=[draw,rectangle,rounded corners=2mm]
  \tikzstyle{bddleaf}=[]
  \tikzstyle{trans}=[->,>=stealth']
  \tikzstyle{translow}=[->,>=stealth']
  \tikzstyle{stick}=[-,>=stealth']
  \tikzstyle{hidtrans}=[]
  \tikzstyle{ark}=[]
  \tikzstyle{blueark}=[fill=white,opacity=1]
  \tikzstyle{redark}=[fill=red,opacity=0.6]

  \tikzstyle{outp}=[scale=0.75,fill=black!30,inner sep=0.6mm]

  \tikzstyle{bddnodex}=[bddnode,inner sep=1mm]

  % NODES
  \node[bddnodex] (p) {$q_0$};
  \node (pp) at ($(p)+(-0.7mm,-0.1mm)$) {};
  \node (ppp) at ($(p)+(+0.25mm,-0.6mm)$) {};
  \node[right of=p,xshift=-10mm] (root) {};
  \node[bddnodex,below left of=p,yshift=-5mm] (q+) {$q_1$};
  \node[bddnodex,below of=p,yshift=0.5mm] (q0) {$q_2$};
  \node[bddnodex,below right of=p,yshift=-5mm] (q-) {$q_3$};
  \node[bddnodex,below left of=q+,yshift=-5mm] (r1) {$q_4$};
  \node[bddnodex,below of=q+,yshift=0mm] (r2) {$q_5$};
  \node[bddnodex,below right of=q-,yshift=-5mm] (r4) {$q_8$};
  \node[bddnodex,below of=q-,yshift=0mm] (r3) {$q_7$};
  \node[bddnodex,below right of=q+,yshift=-5mm] (r0) {$q_6$};

  \node[bddleaf, below of=r1,yshift=11mm] (r1a0) {$\frac{1}{\sqrt{2}}$};
  \node[bddleaf, below of=r2,yshift=11mm] (r2a0) {$\frac{-1}{\sqrt{2}}$};
  \node[bddleaf, below of=r0,yshift=11mm] (r0a) {$0$};
  \node[bddleaf, below of=r3,yshift=11mm] (r3a0) {$\frac{i}{\sqrt{2}}$};
  \node[bddleaf, below of=r4,yshift=11mm] (r4a0) {$\frac{-i}{\sqrt{2}}$};

  \draw (p) coordinate[xshift=-5mm,yshift=-5mm] (pa);
  \draw (p) coordinate[xshift= 5mm,yshift=-5mm] (pb);

  \draw (q+) coordinate[xshift=0mm,yshift=-6mm] (q+a);
  \draw (q-) coordinate[xshift=0mm,yshift=-6mm] (q-a);

  \draw (q0) coordinate[xshift=0mm,yshift=-6mm] (q0a);

  %%%%%%%%%%%%%%% transition %%%%%%%%%%%%%%%%%%%
  \draw[trans,-] (pp)
    to
    (pa);
  \draw[trans,-] (ppp)
    to
    (pb);
  \draw[trans] (pa)
    to[bend right=10]
    node[pos=0,left,xshift=-1mm,yshift=1mm] {}
    coordinate[pos=0.45] (pa1)
    (q+);
  \draw[translow] (pa)
    to[bend left=10]
    coordinate[pos=0.45] (pa2)
    (q0);
  \draw[translow]
    (pb)
    to[bend right=10]
    coordinate[pos=0.45] (pb1)
    (q0);
  \draw[trans]
    (pb)
    to[bend left=10]
    coordinate[pos=0.45] (pb2)
    (q-);

  \node at (pa) [xshift=-1mm,yshift=-6mm] {$\{1\}$};
  \node at (pb) [xshift=+1mm,yshift=-6mm] {$\{2\}$};

  %%%%%%%%%%%%%%% transition %%%%%%%%%%%%%%%%%%%
  \draw[translow] (q+a)
    to[bend right]
    coordinate[pos=0.5] (q+a1)
    (r1);

  \draw[trans] (q+) to 
    node[pos=0.9,right,xshift=0mm,yshift=1mm] {}
    (q+a)
    to
    coordinate[pos=0.6] (q+a2)
    (r2);

  \node at (q+a) [xshift=-5mm,yshift=-5mm] {$\{1\}$};
  \node at (q-a) [xshift=+5mm,yshift=-5mm] {$\{1\}$};

  \draw[translow] (q-a) to 
    node[pos=0.9,left,xshift=0mm,yshift=1mm] {}
    (q-a)
    to
    coordinate[pos=0.6] (q-a1)
    (r3);

  \draw[trans] (q-) to
    (q-a)
    to[bend left]
    coordinate[pos=0.5] (q-a2)
    (r4);

  %%%%%%%%%%%%%%% transition %%%%%%%%%%%%%%%%%%%
  \draw[translow] (q0a)
    to[bend right=50]
    coordinate[pos=0.6] (q1a1)
    ([xshift=-1.8mm, yshift=+2mm]r0);

  \draw[trans] (q0) to 
    node[pos=0.9,left,xshift=0mm,yshift=1mm] {}
    (q0a)
    to[bend left=50]
    coordinate[pos=0.6] (q1a2)
    ([xshift=+2.3mm, yshift=-1.1mm]r0);

  \node at (q0a) [xshift=0mm,yshift=-5mm] {$\{1\}$};

  %%%%%%%%%%%%%%% transition %%%%%%%%%%%%%%%%%%%
  \draw[trans] (root) to (p);
  \draw[stick] (r1) to node[left,xshift=-1mm] {$\{1\}$} (r1a0);
  \draw[stick] (r2) to node[left,xshift=-1mm] {$\{1\}$} (r2a0);
  \draw[stick] (r0) to node[left, xshift=-1mm,yshift=0mm] {$\{1\}$} (r0a);
  \draw[stick] (r3) to node[left,xshift=-1mm] {$\{1\}$} (r3a0);
  \draw[stick] (r4) to node[left,xshift=-1mm] {$\{1\}$} (r4a0);
\end{tikzpicture}
}
    \label{fig:LSTA}
}
}
\caption{Representing a set of quantum states with an LSTA}
\label{fig:Representing a set of quantum states with an LSTA}
\end{figure}
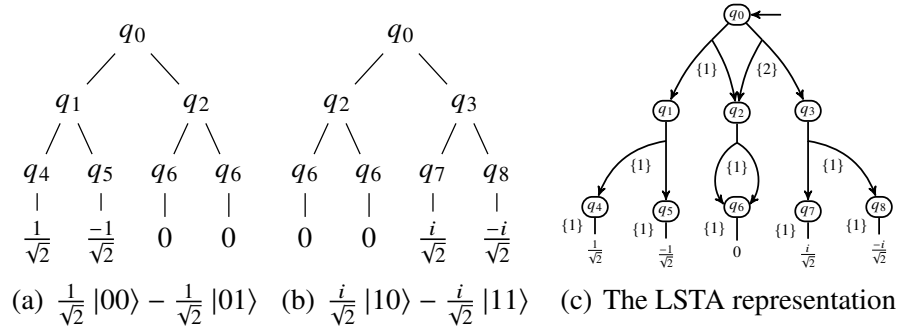
\end{example}

\subsubsection{Binary Operations.}
In this work, the \emph{set union} and \emph{tensor product} operations on LSTAs, which correspond to the semantics of language operations, serve as the fundamental building blocks for constructing the target LSTA. The upper bounds on the sizes of the constructed LSTAs, which are keys to the desired complexity, are summarized in the following theorem (see the \hyperlink{pfBinaryOperationComplexity}{proof} in \Cref{app:bin}).

\begin{restatable}{theorem}{binaryOperationComplexity}\label{theorem:Binary Operations}
Given two LSTAs $\mathcal{A}$ and $\mathcal{B}$ over $\mathbb{K}$ where $\mathcal L(\mathcal A)$ and $\mathcal L(\mathcal B)$ contain $n$-qubit and $m$-qubit states, respectively, there exists a set union operation (denoted by $\mathcal{A} \sqcup \mathcal{B}$) and a tensor product operation (denoted by $\mathcal{A} \otimes \mathcal{B}$), both yielding valid LSTAs over $\mathbb{K}$.
These operations satisfy the semantic properties $\mathcal{L}(\mathcal{A} \sqcup \mathcal{B}) = \mathcal{L}(\mathcal{A}) \cup \mathcal{L}(\mathcal{B})$ and $\mathcal{L}(\mathcal{A} \otimes \mathcal{B}) = \mathcal{L}(\mathcal{A}) \otimes \mathcal{L}(\mathcal{B})$.
Regarding the automaton size, the set union is bounded by $|\mathcal{A} \sqcup \mathcal{B}| \le |\mathcal{A}| + |\mathcal{B}|$ and the tensor product is bounded by $|\mathcal{A} \otimes \mathcal{B}| \le |\mathcal{A}| + N_{\text{leaves}}(\mathcal{A}) \cdot |\mathcal{B}|$, where $N_{\text{leaves}}(\mathcal{A})$ denotes the number of distinct amplitude values at the leaves of $\mathcal{A}$.
\end{restatable}

%%%%%%%%%%%%%%%%%%%%%%%%%%%%%%%%%%%%%%%%%%%%%%%%%%%%%%%%%%%%%%%%%%
\section{Specification Language}\label{sec:Specification Language}
\newcommand{\expr}[0]{\mathit{expr}}
\newcommand{\tset}[0]{\mathit{tset}}
\newcommand{\pset}[0]{\mathit{pset}}
\newcommand{\uset}[0]{\mathit{uset}}
\newcommand{\setq}[0]{\mathit{set}}
\newcommand{\setP}[0]{\mathit{setP}}
\newcommand{\setV}[0]{\mathit{setV}}
\newcommand{\setQ}[0]{\mathit{setQ}}
\newcommand{\cpxcons}[0]{\mathit{CCons}}
\newcommand{\pint}[0]{\mathit{N}}
\newcommand{\diracs}[0]{\mathit{diracs}}
\newcommand{\varcons}[0]{\mathit{varcons}}
\newcommand{\dirac}[0]{\mathit{dirac}}
\newcommand{\term}[0]{\mathit{term}}
\newcommand{\termV}[0]{\mathit{termV}}
\newcommand{\complexq}[0]{\mathit{C}}
\newcommand{\varcon}[0]{\mathit{varcon}}
\newcommand{\var}[0]{\mathit{V}}
\newcommand{\ineqs}[0]{\mathit{ineqs}}
\newcommand{\ineq}[0]{\mathit{ineq}}
\newcommand{\cstr}[0]{\mathit{CStr}}
\newcommand{\vstr}[0]{\mathit{VStr}}

This section formally introduces the specification language used to describe \textit{assertions} (i.e., sets of quantum states), which serve as preconditions and postconditions in our verification framework. Our design aligns with the Dirac notation and standard set representation, ensuring familiarity and expressiveness. The syntax is shown in \cref{fig:syntax}, with $\expr$ as the start symbol. We define two disjoint sets of variable names for later use: $V_s$ for binary string variables (used in basis states) and $V_c$ for complex variables. 

%alt={A set of Backus-Naur Form syntax rules for the specification language. The grammar is divided into two logical parts. The first part defines set-level expressions: 'expr' constructs unions of 'tset' components; 'tset' constructs tensor products of 'pset' components; 'pset' allows tensor powers of 'uset'; 'uset' constructs set unions of 'set' components; and 'set' encapsulates quantum states called 'diracs', optionally conditioned on constraints called 'varcons'. The second part defines the quantum states: 'diracs' is a comma-separated list of 'dirac' terms; 'dirac' is a linear combination of 'term' constructs; 'term' represents a quantum basis state or a superposition over constraints with an amplitude alpha; 'varcons' is a comma-separated list of 'varcon' conditions; and 'varcon' specifies length constraints, variable equalities, or variable inequalities.}
\begin{figure}
    \begin{minipage}{0.44\textwidth}
    \begin{align*}
    \expr \Coloneqq\ & \tset \mid \bigcup_{\cpxcons} \tset\\
    \tset \Coloneqq\ & \pset \mid \tset \otimes \pset\\
    \pset \Coloneqq\ & \uset \mid \uset^\pint\\
    \uset \Coloneqq\ & \setq \mid \uset \cup \setq\\
    \setq \Coloneqq\ & \set{\diracs} \mid \set{\diracs : \varcons}
    \end{align*}
    \end{minipage}
    \begin{minipage}{0.55\textwidth}
    \begin{align*}
    \diracs \Coloneqq\ & \dirac \mid \dirac,\ \diracs\\
    \dirac \Coloneqq\ & \term \mid \dirac + \term\\
    \term \Coloneqq\ & \alpha \ket{\vstr} \mid \alpha \sum_{\varcons}\ \ket{\vstr}\\
    \varcons \Coloneqq\ & \varcon \mid \varcon,\ \varcons\\
    \varcon \Coloneqq\ & |\var|=\pint \mid \var\ne\var \mid \var\ne\cstr \mid \var=\cstr
    \end{align*}
    \end{minipage}
    \caption{Syntax of the specification language. The grammar comprises six terminal categories: natural number $\pint \in \mathbb{N}$; constant binary string $\cstr \in \{0,1\}^+$; complex polynomial $\alpha \in \mathbb{C}[V_c]$; binary string variable $\var \in V_s$; constraint $\cpxcons$, defined as a quantifier-free first-order formula in nonlinear arithmetic over the real part and imaginary part of variables in $V_c$; and basis string pattern $\vstr \in (\{0,1\}\cup V_s \cup \overline{V_s})^+$, where $\overline{V_s} = \{\overline{v} \mid v\in V_s\}$ collects bit-complemented variables.}
    \label{fig:syntax}
\end{figure}

The grammar is designed to possess three key features for quantum program verification:
\begin{itemize}
    \item \textbf{Modular Construction:} Two fundamental operations allow for the construction of complex sets from simpler components: the standard set union ($\cup$) to aggregate states and the tensor product ($\otimes$) to compose quantum subsystems. Our syntax assigns $\otimes$ the lowest precedence among all other operators except $\bigcup_{\cpxcons}$ to make $\otimes$ act as a natural structural delimiter that separates distinct subsystems in our algorithm. For instance, $A \cup B \otimes C$ is interpreted as $(A \cup B) \otimes C$.
    \item \textbf{Symbolic Representation:} Beyond concrete values, the grammar admits symbolic variables in both amplitudes (via $V_c$) and basis states (via $V_s$). The variables in $V_c$ enable the specification of infinite sets of states. The variables in $V_s$ can be used along with the summation ($\sum$) to describe superpositions compactly. This allows a single $\term$ to represent a linear combination of exponentially many basis states without explicit enumeration. For instance, $\displaystyle\bigg\{\dfrac{1}{\sqrt7}\sum_{\substack{j\ne i}} \ket j : |i|=3\bigg\}$ represents a set of 
    normalized uniform superpositions where exactly one basis state $\ket{i}$ is excluded from the full basis.
    \item \textbf{Constraint-Based Specification:} The terminal $\cpxcons$ and nonterminal $\varcons$ enable precise control over the valid state space by applying constraints to variables. For instance, $\displaystyle\bigcup_{|\alpha|^2 > 0.8}\Big\{ \alpha\ket w + \beta\sum_{\substack{i\neq w}}\ket i \Big\}$ for some $n \ge 2$ and marked item $w\in\set{0,1}^n$ can be used as a postcondition of Grover's algorithm.
\end{itemize}

In addition, the grammar provides specific constructs to facilitate compact specifications. We denote the $N$-fold tensor product of a set $S$ as $S^N$, inductively defined by $S^1 = S$ and $S^N = S^{N-1} \otimes S$. This is particularly useful for describing uniform registers, such as an $N$-qubit zero state $\set{\ket{0}}^N$. Furthermore, the constraint $|\var|=\pint$ is used to explicitly define the domain of a variable $\var$ as $\{0,1\}^\pint$. Example usage is provided in \Cref{sec:usecases}. % Ironically, our algorithm terminates if assertions are not tensor aligned, so in practice, we still usually write out the whole string instead of in the tensor product. For now, I cannot think of a scenario where the tensor product is necessary. One possible improvement of the implementation is to merely expand tensor products, but it is not quite easy...

Syntactic correctness does not guarantee semantic validity. To be considered \emph{well-formed}, an \emph{assertion} generated by the grammar must satisfy the following:
\begin{enumerate}[(1)]
    \item \textbf{Unambiguous Variable Length:} For every binary string variable $\var$ present in a $\vstr$, its \emph{length} $\pint$ (i.e., the number of qubits) must be uniquely determined, either explicitly via the constraint $|\var|=\pint$ or implicitly inferred from constraints such as $\var\ne V'$, $\var\ne\cstr$, or $\var=\cstr$, provided that $V'$ or $\cstr$ has a known length of $\pint$.
    \item \textbf{Length Consistency:} The number of qubits must be consistent across all quantum states within a single $\uset$, as well as between the operands of inequalities.
    \item \textbf{No Redundant Summation Variables:} Within each $\term$, every \emph{iterating} variable under the summation (if any) must be present in $\vstr$. Otherwise, the summation results in unintended amplitude scaling, which is undesirable. For instance, $\displaystyle\bigg\{\dfrac{1}{\sqrt7}\sum_{\substack{j\ne i}} \ket j : |i|=3\bigg\}$ has no redundant summation variables because $\vstr$ contains the only iterating variable $j$, but $\displaystyle\bigg\{\dfrac{1}{\sqrt7}\sum_{\substack{j\ne i,\ k\ne i,\ |\ell|=2}} \ket j : |i|=3\bigg\}$ has two redundant variables $k$ and $\ell$, which results in an unintended scaling factor of $(2^3-1)\times 2^2 = 28$.
\end{enumerate}

\section{Use Cases}
\label{sec:usecases}
%Below, we will provide examples of the verification problems and the corresponding specifications in our language.

\subsection{Oracle-Based Algorithms}\label{sec:oracle}

\figBVcircuit

An \emph{oracle circuit} is a black-box circuit that encodes a function and
enables quantum algorithms to query information in a single step. In
Grover's search~\cite{Grover96}, the oracle implements
$f(x)\colon \bool^n \to \bool$, returning $1$ on the marked solution $x$ and $0$
otherwise; in Bernstein--Vazirani (BV)~\cite{BernsteinV93}, it encodes a secret
bit string.
To verify such algorithms for \emph{all} oracles, we use a \emph{parameterized}
oracle circuit whose behavior is determined by input qubits via controlled
gates, and compose it with the circuit under verification.

For BV, the composed circuit has $7$ qubits (\cref{fig:bv-circuit}): the
highlighted block $\bluelab{\ \ }$ is the oracle, and the rest is the
implementation. We treat the secret as part of the input: qubits $s_1,s_2,s_3$
parameterize the oracle, the remaining 0-qubits are workspace, and the last qubit
is an ancilla. We prove correctness by showing
$\{\ket{s0001} : |s|=3\} \Rightarrow \{\ket{ss1} : |s|=3\}$, i.e., for every
secret string $s$, the BV circuit outputs $s$.

\subsection{Amplitude Amplification}\label{sec:amplitude}

When verifying an amplitude amplification algorithm (such as
Grover's search~\cite{Grover96}), we can use variables~$v_h$ and~$v_\ell$ as
amplitudes and describe the relation between the variables before and after the circuit evolution using a global constraint.
For a Grover iteration circuit where the marked state is $\ket{111}$ on the first three qubits and the remaining three qubits serve as ancilla, we verify its correctness using the following specification.
\[
\scalebox{0.9}{$\displaystyle
\bigcup_{\substack{\text{imag}(v_h)=0,\ \text{real}(v_h)>0,\\\text{imag}(v_\ell)=0,\ \text{real}(v_\ell)>0,\\7v_\ell > v_h}}\{v_h\ket{111001}+v_\ell\sum_{\substack{i \ne 111}}\ket{i001}\} \Rightarrow \bigcup_{\substack{\text{imag}(v'_h)=0,\\\text{imag}(v'_\ell)=0,\\|v'_h| > |v_h|}}
\{v'_h\ket{111001}+v'_\ell\sum_{\substack{i \ne 111}}\ket{i001}\}
$},
\]
where $P \Rightarrow Q$ means $P$ is the \emph{precondition} and $Q$ is the \emph{postcondition}.
The language also allows us to specify the property that a complete Grover's circuit has ${>}80\,\%$ probability of finding the marked state as follows.
\[\{\ket{000001}\} \Rightarrow 
\bigcup_{|v_h|^2\,>\,0.8}\{v_h\ket{111001}+v_\ell\sum_{\substack{i \ne 111}}\ket{i001}\}\]
 
\subsection{Compound Multi-Control Quantum Gates}\label{sec:multi-control}

\figMCToffoli

Quantum hardware typically supports only a limited gate set, so implementing an
unsupported gate often requires decomposing it into a sequence of native gates.
For example, an $n$-controlled Toffoli gate is usually realized using standard
Toffoli gates. In \cref{fig:cccx}, the qubits
$c_1,\,\ldots,\,c_5$ are controls, the intermediate $\ket{0}$ registers
are ancillas, and the final qubit $\ket t$ is the target. We verify correctness
against the following specifications:\vspace{3pt}
\scalebox{0.83}{
\begin{tabular}{r @{\;} c @{\;} l}
    $\{\ket{c00000} : c=11111\}$ & $\Rightarrow$ & $\{\ket{c00001} : c=11111\}$ \\
    $\{\ket{c00001} : c=11111\}$ & $\Rightarrow$ & $\{\ket{c00000} : c=11111\}$ \\
    $\{\ket{c00000} : c\neq11111\}$ & $\Rightarrow$ & $\{\ket{c00000} : c\neq11111\}$ \\
    $\{\ket{c00001} : c\neq11111\}$ & $\Rightarrow$ & $\{\ket{c00001} : c\neq11111\}$
\end{tabular}
}

\section{From Specification to Automata}\label{sec:From Specification to Automata}

\subsection{Overview}\label{sec:overview}
This section details the algorithm that translates a list of input assertions into their corresponding automata, with qubits reordered strategically. The procedure employs a divide-and-conquer strategy to reduce construction complexity. First, we preprocess the input to \textbf{ensure the variable-boundary-aligned representation} (\cref{sec:preprocessing}). Next, we perform a two-stage transformation: a \textbf{high-level variable reordering} (\cref{sec:variablereorderingandtensorproducttransformation}) followed by a \textbf{low-level qubit reordering} (\cref{sec:qubitorderingandtensorproducttransformation}). These steps aim to transform assertions into tensor products of smaller set components with qubits reordered, constituting a key contribution of this work. Finally, we \textbf{construct compact automata} for these components and recompose them using set union and tensor product operations to yield the final automaton (\cref{sec:Compact-LSTA-Construction}).

\subsection{Ensuring Variable-Boundary-Aligned Representation}\label{sec:preprocessing}

To facilitate the core translation algorithm, we first process the raw input assertions into variable-boundary-aligned representations through four steps.

\begin{enumerate}
    \item \hypertarget{Canonicalization}{\textbf{Canonicalization (Syntactic Sugar Elimination and Variable Renaming):}} To streamline the core algorithm steps, we canonicalize the input representation in advance. We first eliminate two forms of syntactic sugar through rewriting. Specifically, a tensor power $\uset^\pint$ is expanded into the $N$-fold tensor product $\underbrace{\uset \otimes \cdots \otimes \uset}_{\pint}$, and a $\setq$ of the form $\set{\dirac_1,\ \dirac_2,\ \ldots,\ \dirac_k : \varcons}$ is split into the $k$-fold union $\set{\dirac_1 : \varcons} \cup \set{\dirac_2 : \varcons} \cup \cdots \cup \set{\dirac_k : \varcons}$. This rule applies analogously to sets without $\varcons$. After that, we perform alpha-renaming to ensure that distinct variables are assigned unique names, thereby streamlining the dependency analysis in \Cref{sec:variablereorderingandtensorproducttransformation}. These transformations yield a form that is free of syntactic sugar and ensures variable uniqueness.
    
    \item \textbf{Tensor Alignment Check:} We verify that all assertions after canonicalization contain the same number of tensor product operators ($\otimes$). If this check passes, we further decompose each assertion into a sequence of $\uset$ constructs delimited by these operators. By defining the collection of the $i$-th $\uset$ from all assertions as the $i$-th \emph{tensor segment}, we then verify that within each segment, all quantum states possess the same qubit length. \emph{Any mismatch in segment count or qubit length terminates the process.}

    \item \textbf{Variable Alignment Check:} We verify that the boundaries of all variables are aligned. Operationally, this is implemented by checking that any two variables occupy either disjoint or identical qubit intervals. \emph{Any violation aborts the process.}

    \item \textbf{Constant Abstraction:} Upon passing the alignment checks, we compute the \emph{global partition}, a set of disjoint qubit intervals that partitions the range $[1,\ n+1)$ where $n$ denotes the number of qubits in all assertions. The partition respects the boundaries of all variables, which means each interval $[a, b)$ occupied by a variable contributes exactly one element to the partition. To maintain structural consistency, any constant binary string spanning multiple intervals is sliced to match these boundaries. Each resulting slice is then abstracted into a fresh variable $V$ bound to its corresponding constant value $\cstr$ via inserted equalities $\var=\cstr$ under the summation, ensuring that every rewritten $\vstr$ aligns with the global partition at the variable level.
\end{enumerate}

Following this preprocessing phase, the original $\setq$ constructs are transformed into restricted $\setP$ constructs, in which every $\vstr$ is free of constant binary digits.

\begin{example}[Preprocessing Pipeline]
Consider an input of two assertions.
\[
\begin{aligned}
\mathcal{E}_1 &= \set{\ \ket{i\, 0\, 0}\ :\ |i|=2\ } \otimes \set{\ket{0}} \cup \set{\ket{1}}^2 \\
\mathcal{E}_2 &= \set{\ \ket{00\, 0\, i},\ \ket{11\, 1\, i}\ :\ |i|=1\ } \otimes \set{\ket{0}} \otimes \set{\ket{0}}
\end{aligned}
\]

\textbf{Step 1 (Canonicalization):}
We rewrite $\mathcal{E}_1$ to eliminate the tensor power and rewrite $\mathcal{E}_2$ to eliminate the comma-separated list. Additionally, we resolve the naming conflict by alpha-renaming the variable $i$ in $\mathcal{E}_2$ to $j$ in the first term and $k$ in the second term.
\[
\begin{aligned}
\mathcal{E}_1' &= \underbrace{\set{\ \ket{i\, 0\, 0}\ :\ |i|=2\ }}_{\text{Segment 1}} \otimes \underbrace{\set{\ket{0}} \cup \set{\ket{1}}}_{\text{Segment 2}} \otimes \underbrace{\set{\ket{0}} \cup \set{\ket{1}}}_{\text{Segment 3}} \\
\mathcal{E}_2' &= \underbrace{\set{\ \ket{00\, 0\, j}\ :\ |j|=1\ } \cup \set{\ \ket{11\, 1\, k}\ :\ |k|=1\ }}_{\text{Segment 1}} \otimes \underbrace{\set{\ket{0}}}_{\text{Segment 2}} \otimes \underbrace{\set{\ket{0}}}_{\text{Segment 3}}
\end{aligned}
\]

\textbf{Step 2 (Tensor Alignment Check):} 
We decompose both canonicalized assertions into segments delimited by tensor product operators. Since $\mathcal{E}_1'$ and $\mathcal{E}_2'$ both have three segments, the segment counts pass. We proceed to the qubit lengths. In Segment 1, $\mathcal{E}_1'$ terms have length $2+1+1=4$ and $\mathcal{E}_2'$ term has length $3+1=4$. In Segments 2 and 3, both $\mathcal{E}_1'$ terms and the $\mathcal{E}_2'$ term have length $1$, so the qubit lengths also pass.

\textbf{Step 3 (Variable Alignment Check):}
We check the variable boundaries for Segment 1. In $\mathcal{E}_1'$, variable $i$ occupies $[1, 3)$. In $\mathcal{E}_2'$, the renamed variables $j$ and $k$ occupy $[4, 5)$. Since all intervals are either disjoint or identical, the variable boundaries are consistent across all segments and hence pass this check.

\textbf{Step 4 (Constant Abstraction):}
The global partition $\mathcal{I}$ across all three segments is $\set{\underbrace{\overbrace{[1, 3)}^{\text{Slot 1}}, \overbrace{[3, 4)}^{\text{Slot 2}}, \overbrace{[4, 5)}^{\text{Slot 3}}}_{\text{Segment 1}}, \underbrace{\overbrace{[5, 6)}^{\text{Slot 4}}}_{\text{Segment 2}}, \underbrace{\overbrace{[6, 7)}^{\text{Slot 5}}}_{\text{Segment 3}}}$. We slice and abstract constants to match $\mathcal{I}$.
\[
\begin{aligned}
\mathcal{E}_1'' = \;& \underbrace{\set{\ \sum_{\substack{a=0,\\b=0}}\ket{iab}:|i|=2\ }}_{\text{Segment 1}} \otimes \underbrace{\set{\ \sum_{g=0}\ket{g}\ } \cup \set{\ \sum_{h=1}\ket{h}\ }}_{\text{Segment 2}} \otimes \underbrace{\set{\ \sum_{m=0}\ket{m}\ } \cup \set{\ \sum_{n=1}\ket{n}\ }}_{\text{Segment 3}}\\
\mathcal{E}_2'' = \;& \underbrace{\set{\ \sum_{\substack{c=00,\\d=0}}\ket{cdj}:|j|=1\ } \cup \set{\ \sum_{\substack{e=11,\\f=1}}\ket{efk}:|k|=1\ }}_{\text{Segment 1}} \otimes \underbrace{\set{\ \sum_{\ell=0}\ket{\ell}\ }}_{\text{Segment 2}} \otimes \underbrace{\set{\ \sum_{p=0}\ket{p}\ }}_{\text{Segment 3}}
\end{aligned}
\]

The final rewritten assertions $\mathcal{E}_1''$ and $\mathcal{E}_2''$ consist purely of variable-boundary-aligned terms, ready for transformation. In this aligned form, we refer to each interval in the global partition as a \emph{slot} to highlight the intuition that each slot can be occupied by exactly one variable. In this case, there are five slots numbered from 1 to 5 in this specification, where Segment 1 occupies the first three slots, Segment 2 occupies the fourth slot, and Segment 3 occupies the fifth slot.
\end{example}

\subsection{Variable-Level Reordering and Tensor Product Transformation}\label{sec:variablereorderingandtensorproducttransformation}
This section functions as a high-level complexity reduction strategy. It analyzes variable-level dependencies among slots. Two slots are considered dependent if they are occupied by the same variable or by variables constrained by an inequality. 
Based on this dependency analysis, the algorithm partitions all slots into disjoint independent subsets. It then transforms each $\setP$ construct in the assertion into a tensor product of smaller ones according to the partition. This is the first key technique for cost reduction.

This strategy works best when slots are mutually independent. In this case, the transformation reduces the construction complexity from exponential to linear in the number of variables. In the following example, since the two slots on the left-hand side are independent, this phase transforms the expression as follows:

\[
\Bigg\{\sum_{\substack{|i|=5,\ |j|=5}} \ket{ij} \Bigg\}
\;\longrightarrow\;
\Bigg\{\sum_{|i|=5} \ket{i} \Bigg\} \otimes \Bigg\{\sum_{|j|=5} \ket{j} \Bigg\}
\]
This allows subsequent steps to process $i$ and $j$ separately, reducing the enumeration count from $2^{|i|+|j|}$ ($1024$ states) to $2^{|i|} + 2^{|j|}$ ($32+32=64$ states). Conversely, such a reduction is unachievable when all slots are mutually dependent. In either case, the complexity is guaranteed not to increase. % I am afraid that simply reducing a factor in complexity analysis is not enough to guarantee this.

\hypertarget{slot reordering}{\subsubsection{Slot Reordering.}}\label{sec:slot-reordering}
We model the slot dependencies using an undirected graph $G = (\mathcal{K}, E)$, where $\mathcal{K}$ consists of slot indices. An edge $(i, j) \in E$ (with $i \ne j$) exists if and only if there exist variables $u$ at slot $i$ and $v$ at slot $j$ that satisfy the \textbf{recurrence condition} (i.e., refer to the same variable instance) or the \textbf{inequality condition} (i.e., are constrained by the inequality $u \ne v$). After building this graph, we compute the connected components. To ensure a deterministic transformation structure, we first arrange these components in ascending order of their minimum slot indices. Then, we arrange the vertices in each component into a list and sort each list in ascending order. Finally, we obtain a new total slot order by concatenating these ordered lists.

\begin{example}[Slot Reordering]\label{ex:variable-level-one}
Consider two $\setP$ constructs $S_A$ and $S_B$ within a tensor segment, sharing the same 7-slot structure.
\[
S_A = \Bigg\{\, \alpha_{1A}\sum_{\substack{|a|=1,\ |b|=1, \\ |c|=2,\ |d|=2, \\ e=0,\ a\ne b}}\ket{abcxwde} + \alpha_{2A}\sum_{\substack{|f|=1,\ |g|=1, \\ |h|=2,\ |j|=2, \\ |k|=1}}\ket{fghijwk} : \substack{|i|=1, \\ |w|=2, \\ |x|=1} \,\Bigg\},
\]
\[
S_B = \Bigg\{\, \alpha_{1B}\sum_{\substack{|l|=1,\ |q|=2,\\|m|=2,\ |n|=1, \\ y=0,\ p\ne q}}\ket{lupyqmn} + \alpha_{2B}\sum_{\substack{|o|=1,\ |r|=1, \\ |s|=2,\ |t|=1, \\ |v|=2}}\ket{orstzvu} : \substack{|p|=2, \\ |u|=1, \\ |z|=2, \\ p\ne z} \,\Bigg\}.
\]

We construct the dependency graph for slot indices $\{1, 2, \dots, 7\}$. In $S_A$, the recurrence condition arises from variable $w$ occupying slot 5 in the first term and slot 6 in the second term. This creates edge $(5, 6)$. The inequality condition arises from $a \ne b$ in the first term, which creates edge $(1, 2)$. In $S_B$, the recurrence condition arises from variable $u$ occupying slot 2 in the first term and slot 7 in the second term. This creates edge $(2, 7)$. The inequality condition arises from $p \ne z$ in the set predicate and $p \ne q$ in the first term, which both create edge $(3, 5)$. The resulting dependency graph and its connected components are illustrated in \Cref{fig:slotdependencygraph}. The resulting new total slot order is therefore $\big[[1,2,7], [3,5,6], [4]\big]$.

%alt={A slot dependency graph containing seven nodes numbered 1 to 7, divided into three disjoint connected components enclosed in dashed boxes. The first component contains nodes 1, 2, and 7. Node 1 is connected to node 2 with a dashed red line labeled "a not equal to b" indicating an inequality constraint. Node 2 is connected to node 7 with a solid blue line labeled "var u" indicating a recurrence dependency. The second component contains nodes 3, 5, and 6. Node 3 is connected to node 5 with one dashed red line with two labels: "p not equal to z" above the line and "p not equal to q" below the line, representing inequality constraints. Node 5 is connected to node 6 with a solid blue line labeled "var w", representing a recurrence dependency. The third component contains only node 4 with no connections.}
\begin{figure}[ht]
    \centering
    \begin{tikzpicture}[
        node distance=1cm,
        slot/.style={circle, draw=black, thick, minimum size=8mm, font=\bfseries},
        recurrence/.style={-, thick, blue!80!black},
        inequality/.style={-, thick, dashed, red!80!black},
        group box/.style={draw=gray, rounded corners, dashed, inner sep=5pt}
    ]

        \node[slot] (s1) {1};
        \node[slot] (s2) [right=of s1] {2};
        \node[slot] (s7) [right=of s2] {7};

        \node[slot] (s3) [right=1cm of s7] {3};
        \node[slot] (s5) [right=of s3] {5};
        \node[slot] (s6) [right=of s5] {6};

        \node[slot] (s4) [right=1cm of s6] {4};

        \draw[inequality] (s1) -- node[above, red] {$a\ne b$} (s2);
        \draw[recurrence] (s2) -- node[above, blue] {var $u$} (s7);

        \draw[inequality] (s3) -- node[above, red] {$p\ne z$} (s5);
        \draw[recurrence] (s5) -- node[above, blue] {var $w$} (s6);
        \draw[inequality] (s3) -- node[below, red] {$p\ne q$} (s5);

        \begin{pgfonlayer}{background}
            \node[group box, fit=(s1)(s2)(s7), label=below:\textcolor{gray}{CC$_1$}] {};
            \node[group box, fit=(s3)(s5)(s6), label=below:\textcolor{gray}{CC$_2$}] {};
            \node[group box, fit=(s4), label=below:\textcolor{gray}{CC$_3$}] {};
        \end{pgfonlayer}
    \end{tikzpicture}
    \caption{Dependency graph for slot indices based on $S_A$ and $S_B$. Solid blue lines indicate recurrence dependencies (shared variables), and dashed red lines indicate inequality constraints. The graph reveals three disjoint connected components. In each component, the numbers are arranged in ascending order. In the whole graph, the components are arranged in ascending order of their minimum elements.}
    \label{fig:slotdependencygraph}
\end{figure}
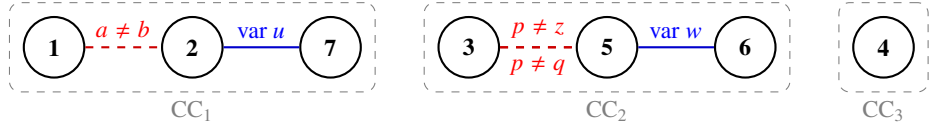
\end{example}

\subsubsection{Tensor Product Transformation.}
Assume a $\setP$ construct $S$ occupies the variable slots reordered by concatenating $k$ lists $L_1, \ldots, L_k$, each corresponding to a connected component from \hyperlink{slot reordering}{slot reordering}. Based on this new slot order, we transform $S$ into a tensor product structure $S_1 \otimes S_2 \otimes \dots \otimes S_k$, where each $S_i$ is obtained by projecting $S$ onto the slots in $L_i$. For clarity, we designate these $S_i$ components as $\setV$ constructs (extending the syntax in \Cref{fig:syntax}) to explicitly identify them as intermediate outcomes of the variable-level transformation. To facilitate this transformation and prevent undesired recombinations of projected terms\footnote{Projected terms are of the form $\tau_m \ket{\vstr}$ or $\tau_m \sum_{\varcons}\ \ket{\vstr}$.} originating from distinct source $\term$ constructs during the final composition, we introduce \emph{tag amplitudes} $\mathcal T = \set{\tau_0,\ \ldots,\ \tau_t}$\footnotemark, where $t$ is the number of $\term$ constructs in $S$. Here, $\tau_0$ is the additive identity and is absorbing for multiplication. These tags satisfy the idempotence and orthogonality properties:
\[
\tau_m \cdot \tau_n =
\begin{cases}
\tau_m & \text{if } m = n, \\
\tau_0 & \text{otherwise.}
\end{cases}
\]
\footnotetext{The domain $\set{\sum_{m\in T} \tau_m}_{T \in 2^{\mathcal T\setminus\set{\tau_0}} \setminus \set{\emptyset}} \cup \set{\tau_0}$, equipped with the defined binary operations, forms a commutative nonunital semiring.}

With tag amplitudes\hide{this definition}, we construct each $\setV$ construct $S_i$ as follows.
\begin{enumerate}
    \item \textbf{VStr Extraction:} For each $\term$ construct, we derive a new string $\vstr2$ from $\vstr$ by selecting the variables located at the slot indices in the order specified by $L_i$, and then replace the original $\vstr$ with this new $\vstr2$.
    \item \textbf{Constraint Filtering:} Regardless of whether constraints appear under the summation or within the set predicate, we retain only those involving the variables present in $\vstr2$. Irrelevant constraints are discarded accordingly.    
    \item \textbf{Amplitude Replacement:} The amplitude $\alpha_m$ of the $m$-th $\term$ is replaced by the tag $\tau_m$\hide{extended version $T(m)$}. This replacement uniquely tags the term, ensuring only the projected terms originating from the same source term are recombined during the tensor product.
\end{enumerate}

The resulting set $S_i$ comprises $t$ instances of a variant construct, denoted as $\termV$. Extending the syntax in \Cref{fig:syntax}, $\termV$ is structurally identical to $\term$, except that the complex amplitude is replaced by a tag amplitude.

\begin{example}[Tensor Product Transformation]\label{ex:variable-level-two} Recall in \Cref{ex:variable-level-one} that the resulting new total slot order is $\big[[1,2,7], [3,5,6], [4]\big]$ and both sets ($S_A$ and $S_B$) are transformed according to this order. For instance, $S_B$ is transformed into $S_{B,1} \otimes S_{B,2} \otimes S_{B,3}$:

\[
S_{B,1} = \Bigg\{\, \tau_1\sum_{\substack{|l|=1, \\ |n|=1}}\ket{lun} + \tau_2\sum_{\substack{|o|=1, \\ |r|=1}}\ket{oru} : |u|=1 \,\Bigg\},
\]
\[
S_{B,2} = \Bigg\{\, \tau_1\sum_{\substack{|q|=2, \\ |m|=2, \\ p\ne q}}\ket{pqm} + \tau_2\sum_{\substack{|s|=2, \\ |v|=2}}\ket{szv} : \substack{|p|=2, \\ |z|=2, \\ p\ne z} \,\Bigg\},\
% \]
% \[
S_{B,3} = \Bigg\{\, \tau_1\sum_{y=0}\ket{y} + \tau_2\sum_{|t|=1}\ket{t} \,\Bigg\}.
\]
% For instance, $S_A$ is transformed into $S_{A,1} \otimes S_{A,2} \otimes S_{A,3}$:
% \[
% S_{A,1} = \Bigg\{\, T(1)\sum_{\substack{|a|=2,\ |b|=2, \\ e=00,\ a\ne b}}\ket{abe} + T(2)\sum_{\substack{|f|=2,\ |g|=2, \\ |k|=2}}\ket{fgk} \,\Bigg\},
% \]
% \[
% S_{A,2} = \Bigg\{\, T(1)\sum_{\substack{|c|=1, \\ |d|=1}}\ket{cwd} + T(2)\sum_{\substack{|h|=1, \\ |j|=1}}\ket{hjw} : |w|=1 \,\Bigg\},
% \]
% \[
% S_{A,3} = \Bigg\{\, T(1)\ket{x} + T(2)\ket{i} : |x|=1,\ |i|=1 \,\Bigg\}.
% \]
\end{example}

\subsection{Qubit-Level Reordering and Tensor Product Transformation}\label{sec:qubitorderingandtensorproducttransformation}
Unlike the preceding variable-level transformation, which treats variables as atomic units, inequalities such as $i \neq j$ can always be resolved at the qubit level via the logical disjunction of $i_k \neq j_k$ across all qubits $k$ (i.e., $i_1\neq j_1 \vee \ldots \vee i_\ell\neq j_\ell$). This observation enables us to further decompose the dependency structure: distinct qubits can be handled in separate constructs, linked only by the accumulated satisfaction status of constraints. By expanding variables into individual qubits, this phase effectively reduces the time complexity from exponential to linear in the number of qubits.

Let $\ell$ denote the qubit length of the variables in a $\setV$ construct $S$ derived from the previous phase. We expand each multi-qubit variable $v$ into a sequence of single-qubit variables $v_1 v_2 \dots v_\ell$. Based on this expansion, we transform $S$ into a qubit-level tensor product structure $S_1 \otimes S_2 \otimes \dots \otimes S_\ell$, where each $S_k$ governs the $k$-th qubit slice. We refer to each component $S_k$ as a $\setQ$ construct (an internal structure not covered in \Cref{fig:syntax}), which utilizes a specialized amplitude form to track the partial satisfaction of constraints (e.g., $i \neq j \iff \bigvee_k (i_k \neq j_k)$).

\subsubsection*{Valuation-Dependent Amplitudes.}
To implement the disjunctive logic while maintaining the accumulated satisfaction status, we introduce \hypertarget{def:valuation-dependent amplitudes}{\emph{valuation-dependent amplitudes}}.
Consider a $\setV$ construct containing $t$ instances of $\termV$. For the $m$-th $\termV$, let $\Phi_m$ denote the set of inequality constraints collected from both the local summation within $\termV$ and the global set predicate of $\setV$. We define the valuation-dependent amplitude as a collection of boolean functions. Formally, let $d$ be a set defined as $d \triangleq \set{f_m}_{m \in T}$ for some subset of term indices $T \subseteq [t]$, where each element $f_m: \Phi_m \to \mathbb{B}$ maps the constraints in $\Phi_m$ to truth values\footnotemark.

The algebraic operations on $d_1 \triangleq \set{f^1_m}_{m \in T_1}$ and $d_2 \triangleq \set{f^2_m}_{m \in T_2}$ are defined as follows:
\begin{itemize}
    \item \textbf{Sum:} $d_1 + d_2 \triangleq \{ f_m \}_{m \in T_1 \cup T_2}$, where $f_m \triangleq f^1_m$ if $m \in T_1$ and $f^2_m$ if $m \in T_2$. It is well-defined because in this work, this operation is applied only when $T_1 \cap T_2 = \emptyset$.
    \item \textbf{Product:} $d_1 \cdot d_2 \triangleq \set{f_m}_{m \in T_1 \cap T_2}$, where $f_m(\phi) \triangleq f^1_m(\phi) \lor f^2_m(\phi), \quad \forall \phi \in \Phi_m$.
\end{itemize}
\footnotetext{The amplitude domain consists of all partial mappings $d = \set{f_m}_{m \in T}$ from subsets of term indices $T \subseteq [t]$ to boolean functions. By extending the sum operation for overlapping indices $m \in T_1 \cap T_2$ as the pointwise disjunction $f_m(\phi) \triangleq f^1_m(\phi) \lor f^2_m(\phi)$, this domain, equipped with the defined binary operations, forms a commutative nonunital semiring.}

With these definitions, we construct each $\setQ$ construct $S_k$ from $S$ as follows:
\begin{enumerate}
    \item \textbf{Qubit Projection and Constraint Relaxation:} We project all variables in $S$ onto their $k$-th qubits (adding subscript $k$). Initially, we relax all constraints, meaning that all single-qubit variables are freely instantiated to values in $\{0, 1\}$.
    \item \textbf{Constraint Evaluation:} The previously disregarded constraints are then integrated into the valuation-dependent amplitudes. Specifically, the original tag amplitude $\tau_m$ is replaced by the singleton set $\{f_m\}$, where the truth value of each constraint in $\Phi_m$ is determined locally by the current assignment of the qubit-level variables.
    \item \textbf{Concrete Expansion:} Finally, each quantum state in $\setQ$ is expanded into a summation of concrete basis states with these newly computed valuation-dependent amplitudes. These concrete expansions are then passed to the automaton construction procedure described in the next section.
\end{enumerate}

\begin{example}[Qubit-Level Construction for $S_{B,2}$ in \Cref{ex:variable-level-two}]\label{ex:Qubit-Level Construction}
Recall that
\[
S_{B,2} = \Bigg\{\, \tau_1\sum_{\substack{|q|=2, \\ |m|=2, \\ p\ne q}}\ket{pqm} + \tau_2\sum_{\substack{|s|=2, \\ |v|=2}}\ket{szv} : \substack{|p|=2, \\ |z|=2, \\ p\ne z} \,\Bigg\}.
\]
Since the qubit length is 2, we construct two qubit slices $S_{B,2}^{(1)}$ and $S_{B,2}^{(2)}$. The compact algebraic form\footnote{The concrete expansion is given in \Cref{ex:Concrete Expansion and LSTA Construction}.} of the $j$-th qubit slice ($j=1, 2$) is given by:
\[
S_{B,2}^{(j)} = \left\{
    \sum_{q_j, m_j} \big\{f_1^{V_1}\big\}\ket{p_j q_j m_j} 
    + 
    \sum_{s_j, v_j} \big\{f_2^{V_2}\big\}\ket{s_j z_j v_j}
    : p_j, z_j
\right\},
\]
where the constraint sets are $\Phi_1 \triangleq \{p\ne z,\ p\ne q\}$ and $\Phi_2 \triangleq \{p\ne z\}$.
For brevity, we denote the valuation functions as $f_m^{V_m}$, where $V_1$ and $V_2$ represent the boolean assignments to the variables $\{p_j, z_j, q_j, m_j\}$ and $\{p_j, z_j, s_j, v_j\}$, respectively. The function bodies are determined locally:
$f_1^{V_1}(p\ne z) \triangleq (p_j\ne z_j)$, 
$f_1^{V_1}(p\ne q) \triangleq (p_j\ne q_j)$, and 
$f_2^{V_2}(p\ne z) \triangleq (p_j\ne z_j)$.
\end{example}

\subsection{Compact LSTA Construction}\label{sec:Compact-LSTA-Construction}
For a set of quantum states in Dirac notation, we construct an LSTA for each quantum state and take the set union of the resulting LSTAs via the LSTA set union operation. Recalling that a quantum state can be represented by a perfect binary tree, we construct the LSTA directly mimicking this tree structure, with all transitions enabled by the singleton choice $\{1\}$. This compact construction holds for any valid amplitude domain $\mathbb{K}$ and yields an upper bound guarantee on the resulting automaton size by employing a bottom-up approach that merges isomorphic subtrees. The result is summarized in the following theorem (see the \hyperlink{pfAtomicSpaceComplexity}{proof} in \Cref{app:space}).

\begin{restatable}{theorem}{atomicSpaceComplexity}\label{thm:atomicSpaceComplexity}
Let $\ket{\psi} = \sum_{\substack{s \in \{0,1\}^n}} a_s\ket{s}$ be an $n$-qubit state. Let $N = |\ \{s \in \{0,1\}^n \mid a_s \neq 0_\mathbb{K}\}\ |$ denote the number of nonzero-amplitude terms. The size of the LSTA $\mathcal{A} = \langle Q, V, \Delta, r\rangle_\mathbb{K}$, constructed via the levelwise procedure detailed below, is bounded by $|\Delta| = O(N \cdot n)$. 
\end{restatable}

The construction for a quantum state in the form of $\sum_{s \in \{0,1\}^n\ :\ a_s \ne 0_\mathbb K} a_s\ket{s}$ proceeds level by level, from the leaves up to the root:
\begin{enumerate}[(1)]
    \item \textbf{Leaf Level:}
    We create a state $q_s$ for each basis state $\ket s$ following a nonzero coefficient $a_s \in \mathbb{K}$, assigning the leaf transition $q_s \xrightarrow{\{1\}} a_s$. Additionally, we construct a default sink state $q_{\bot}^{n}$ with $q_{\bot}^{n} \xrightarrow{\{1\}} 0_\mathbb{K}$ to handle missing terms. Let $Q_n$ denote the set of these explicit states.

    \item \textbf{Internal Levels (Iterate $\ell$ from $n-1$ down to $0$):}
    At each level, we first construct a sink state $q_{\bot}^{\ell}$ with the internal transition $q_{\bot}^{\ell} \xrightarrow{\{1\}} (q_{\bot}^{{\ell+1}}, q_{\bot}^{{\ell+1}})$. Next, identifying the set of active prefixes $X_\ell = \{ x \in \{0,1\}^\ell \mid q_{x0} \in Q_{\ell+1} \lor q_{x1} \in Q_{\ell+1} \}$, we construct a state $q_x$ for each $x \in X_\ell$ with the transition $q_x \xrightarrow{\{1\}} (u_0, u_1)$. Here, the child state $u_b$ (for $b \in \{0,1\}$) is resolved to the explicit state $q_{xb}$ if $q_{xb} \in Q_{\ell+1}$, and defaults to $q_{\bot}^{{\ell+1}}$ otherwise. The set $Q_\ell$ is then updated to include these newly created states.

    \item \textbf{Root Assignment:} The state $q_\epsilon \in Q_0$ is designated as the root state.
\end{enumerate}
It is worth noting that if $a_s \ne 0_\mathbb K$ for all $s\in\set{0,1}^n$ (i.e., full support), the construction of sink states and their associated transitions becomes superfluous and can be skipped.

\begin{example}[Concrete Expansion and LSTA Construction for $S_{B,2}^{(j)}$ in \Cref{ex:Qubit-Level Construction}]\label{ex:Concrete Expansion and LSTA Construction}
Recall
\[
S_{B,2}^{(j)} = \left\{
    \sum_{q_j,\ m_j} \big\{f_1^{V_1}\big\}\ket{p_j q_j m_j} 
    + 
    \sum_{s_j,\ v_j} \big\{f_2^{V_2}\big\}\ket{s_j z_j v_j}
    : p_j, z_j
\right\}\text{ can be expanded into:}
\]
{\fontsize{8.0pt}{9.6pt}\selectfont
\[
\underbrace{\left\{
    \begin{aligned}
        & \Big\{f_1^{\{p\ne z \mapsto \mathsf{F},\ p\ne q \mapsto \mathsf{F}\}}\Big\} (\ket{000} + \ket{001}) \\
        +\ & \Big\{f_1^{\{p\ne z \mapsto \mathsf{F},\ p\ne q \mapsto \mathsf{T}\}}\Big\} (\ket{010} + \ket{011}) \\
        +\ & \Big\{f_2^{\{p\ne z \mapsto \mathsf{F}\}}\Big\} (\ket{000} + \ket{001} + \ket{100} + \ket{101})
    \end{aligned} 
\right\}}_{p_j=0,\ z_j=0}
\cup
\underbrace{\left\{
    \begin{aligned}
        & \Big\{f_1^{\{p\ne z \mapsto \mathsf{T},\ p\ne q \mapsto \mathsf{F}\}}\Big\} (\ket{000} + \ket{001}) \\
        +\ & \Big\{f_1^{\{p\ne z \mapsto \mathsf{T},\ p\ne q \mapsto \mathsf{T}\}}\Big\} (\ket{010} + \ket{011}) \\
        +\ & \Big\{f_2^{\{p\ne z \mapsto \mathsf{T}\}}\Big\} (\ket{010} + \ket{011} + \ket{110} + \ket{111})
    \end{aligned} 
\right\}}_{p_j=0,\ z_j=1}
\]
\[
\cup
\underbrace{\left\{
    \begin{aligned}
        & \Big\{f_1^{\{p\ne z \mapsto \mathsf{T},\ p\ne q \mapsto \mathsf{T}\}}\Big\} (\ket{100} + \ket{101}) \\
        +\ & \Big\{f_1^{\{p\ne z \mapsto \mathsf{T},\ p\ne q \mapsto \mathsf{F}\}}\Big\} (\ket{110} + \ket{111}) \\
        +\ & \Big\{f_2^{\{p\ne z \mapsto \mathsf{T}\}}\Big\} (\ket{000} + \ket{001} + \ket{100} + \ket{101})
    \end{aligned} 
\right\}}_{p_j=1,\ z_j=0}
\cup
\underbrace{\left\{
    \begin{aligned}
        & \Big\{f_1^{\{p\ne z \mapsto \mathsf{F},\ p\ne q \mapsto \mathsf{T}\}}\Big\} (\ket{100} + \ket{101}) \\
        +\ & \Big\{f_1^{\{p\ne z \mapsto \mathsf{F},\ p\ne q \mapsto \mathsf{F}\}}\Big\} (\ket{110} + \ket{111}) \\
        +\ & \Big\{f_2^{\{p\ne z \mapsto \mathsf{F}\}}\Big\} (\ket{010} + \ket{011} + \ket{110} + \ket{111})
    \end{aligned} 
\right\}}_{p_j=1,\ z_j=1}.
\]
} % \fontsize{8.9pt}{10pt}\selectfont

In the above expansion, the superscripts of functions $f_m$ are changed from $V_m$ to the definition body of $f_m$, implying that the explicit valuations of single-qubit variables are no longer required for the subsequent steps.

Take the third case $(p_j=1,\ z_j=0)$ as an example demonstrating the LSTA construction.
\[ 
    \bigg\{f_2^{\{p\ne z \,\mapsto\, \mathsf{T}\}}\bigg\}
\left(\begin{matrix} \ket{000} \\ + \ket{001} \end{matrix} \right) 
+ \bigg\{f_1^{\left\{\substack{p\ne z \,\mapsto\, \mathsf{T} \\ p\ne q \,\mapsto\, \mathsf{T}} \right\}},\ f_2^{\{p\ne z \,\mapsto\, \mathsf{T}\}}\bigg\}
\left(\begin{matrix} \ket{100} \\ + \ket{101} \end{matrix} \right) 
+ \bigg\{f_1^{\left\{\substack{p\ne z \,\mapsto\, \mathsf{T} \\ p\ne q \,\mapsto\, \mathsf{F}} \right\}}\bigg\}
\left(\begin{matrix} \ket{110} \\ + \ket{111} \end{matrix} \right).
\]

According to the construction procedure, the resulting transitions are shown below.
\begin{align*}
\Delta = \bigl\{ 
    & q_{000} \xrightarrow{\{1\}} \bigg\{f_2^{\{p\ne z \,\mapsto\, \mathsf{T}\}}\bigg\}, \ 
    q_{001} \xrightarrow{\{1\}} \bigg\{f_2^{\{p\ne z \,\mapsto\, \mathsf{T}\}}\bigg\}, \\
    & q_{100} \xrightarrow{\{1\}} \bigg\{f_1^{\left\{ \substack{p\ne z \,\mapsto\, \mathsf{T} \\ p\ne q \,\mapsto\, \mathsf{T}} \right\}},\ f_2^{\{p\ne z \,\mapsto\, \mathsf{T}\}}\bigg\}, \ 
    q_{101} \xrightarrow{\{1\}} \bigg\{f_1^{\left\{ \substack{p\ne z \,\mapsto\, \mathsf{T} \\ p\ne q \,\mapsto\, \mathsf{T}} \right\}},\ f_2^{\{p\ne z \,\mapsto\, \mathsf{T}\}}\bigg\}, \\ 
    & q_{110} \xrightarrow{\{1\}} \bigg\{f_1^{\left\{ \substack{p\ne z \,\mapsto\, \mathsf{T} \\ p\ne q \,\mapsto\, \mathsf{F}} \right\}}\bigg\}, \ 
    q_{111} \xrightarrow{\{1\}} \bigg\{f_1^{\left\{ \substack{p\ne z \,\mapsto\, \mathsf{T} \\ p\ne q \,\mapsto\, \mathsf{F}} \right\}}\bigg\}, \ 
    q_{\bot}^3 \xrightarrow{\{1\}} \emptyset, \
    q_{\bot}^2 \xrightarrow{\{1\}} (q_{\bot}^3, q_{\bot}^3), \\ 
    & q_{00} \xrightarrow{\{1\}} (q_{000}, q_{001}), \ 
    q_{10} \xrightarrow{\{1\}} (q_{100}, q_{101}), \ 
    q_{11} \xrightarrow{\{1\}} (q_{110}, q_{111}), \\ 
    & q_{\bot}^1 \xrightarrow{\{1\}} (q_{\bot}^2, q_{\bot}^2), \
    q_{0} \xrightarrow{\{1\}} (q_{00}, q_{\bot}^2), \
    q_{1} \xrightarrow{\{1\}} (q_{10}, q_{11}), \
    q_\epsilon \xrightarrow{\{1\}} (q_0, q_1)
\bigr\}.
\end{align*}

Infer $Q$ from $\Delta$. Then $\langle Q, \emptyset, \Delta, q_\epsilon\rangle$ is the resulting LSTA for the third state. There are four states in $S_{B,2}^{(j)}$, so we repeat the construction procedure for the remaining three states and take the LSTA set union of all four LSTAs to obtain the final LSTA for $S_{B,2}^{(j)}$.
\end{example}

\subsection{Final Assembly of LSTA}
Up to this point, all $\setq$ constructs within the assertions have been broken down into $\setQ$ constructs with their corresponding qubit-level LSTAs $\{\mathbf{M}_Q\}$. First, these LSTAs are combined via the tensor product operation to form variable-level LSTAs $\{\mathbf{M}'_V\}$, following the specified decomposition structure. To filter out invalid assignments, we apply a filter mapping $\mathsf{filter}_f$ to the valuation-dependent amplitudes in each $\mathbf{M}'_V$. We denote $f_m\equiv\mathbf{1}$ if $f_m(\phi) = \mathsf{T}$ for all $\phi \in \Phi_m$. Accordingly, the mapping is defined as $\mathsf{filter}_f(\set{f_m}_{m\in T}) = \sum_{m\in T, f_m\equiv\mathbf{1}}\tau_m$ (retaining terms where all applicable constraints are satisfied), or $\tau_0$ if the set $\{m\in T : f_m\equiv\mathbf{1}\}$ is empty. This effectively replaces leaf transitions $q \xrightarrow{C} e$ with $q \xrightarrow{C} \mathsf{filter}_f(e)$, resulting in LSTAs $\{\mathbf{M}_V\}$ that represent $\setV$ constructs. This mapping serves as a \textit{validation filter}, ensuring that each $\mathbf{M}_V$ encapsulates only valid $\termV$ constructs.
% \footnotetext{We say $f_m\equiv\mathbf{1}$ if $f_m(\phi) = \mathsf{T}, \quad \forall \phi \in \Phi_m$.}

\begin{example}[Application of $\mathsf{filter}_f$ to Amplitudes in \Cref{ex:Concrete Expansion and LSTA Construction}]
The mapping operates on $\Delta$ as follows:
$\mathsf{filter}_f\Big(\Big\{f_2^{\{p\ne z \,\mapsto\, \mathsf{T}\}}\Big\}\Big) = \tau_2$,\ 
$\mathsf{filter}_f\Big(\Big\{f_1^{\left\{ \substack{p\ne z \,\mapsto\, \mathsf{T} \\ p\ne q \,\mapsto\, \mathsf{T}} \right\}},\ f_2^{\{p\ne z \,\mapsto\, \mathsf{T}\}}\Big\}\Big) = \tau_1 + \tau_2$, 
$\mathsf{filter}_f\Big(\Big\{f_1^{\left\{ \substack{p\ne z \,\mapsto\, \mathsf{T} \\ p\ne q \,\mapsto\, \mathsf{F}} \right\}}\Big\}\Big) = \tau_0$, and $\mathsf{filter}_f(\emptyset) = \tau_0$.%\vspace{-2pt}
\paragraph{Remark.} This is only a demonstrating example. The real application will occur right after $\bigotimes_{j} \text{LSTA}(S_{B,2}^{(j)}) \in \set{\mathbf{M}'_V}$ has been constructed to obtain $\text{LSTA}(S_{B,2}) \in \set{\mathbf{M}_V}$.
\end{example}

Subsequently, the LSTAs in $\{\mathbf{M}_V\}$ are again tensored according to the higher-level decomposition structure to produce LSTAs $\{\mathbf{M}'_P\}$. We apply another filter mapping $\mathsf{filter}_\tau$, defined by $\mathsf{filter}_\tau(\sum_{m\in T}\tau_m) = \sum_{m\in T\setminus\set{0}}\alpha_m$ and $0$ if $T = \set{0}$, to the amplitudes in each $\mathbf{M}'_P$. This transforms leaf transitions $q \xrightarrow{C} e$ into $q \xrightarrow{C} \mathsf{filter}_\tau(e)$, yielding the final LSTAs $\{\mathbf{M}_P\}$ for $\setP$ constructs. This second stage ensures only projected terms originating from the same source term are recomposed during the tensor product.

Finally, all LSTAs in $\{\mathbf{M}_P\}$ now possess standard amplitudes in $\mathbb{C}[V_c]$. They are further fused by successively applying LSTA set union and tensor product operations as indicated by the remaining operators in the assertion to obtain the final LSTA $\mathcal{M}$.

As such, each assertion reduces to the form $\bigcup_{\theta\models\cpxcons} \mathcal{L}(\mathcal{M}(\theta))$, where $\theta$ denotes a valuation of complex variables. This formulation signifies that the assertion represents a set of concrete quantum states, each derived from a symbolic tree recognized by $\mathcal{M}$ under a concrete instantiation that satisfies the constraint $\cpxcons$.

\subsection{Space Complexity of LSTA}
In this section, we establish the space complexity of the LSTA constructed from an assertion, detailed in the following theorem\footnote{See the \hyperlink{pfFinalSpaceComplexity}{proof} in \Cref{app:space}.}. We characterize the complexity in terms of the total number of qubits ($L$) and four auxiliary structural parameters: the number of $\term$ constructs ($N_{\text{term}}$) and $\varcon$ constructs ($N_{\text{vc}}$) within each $\setq$ construct, the number of $\dirac$ constructs ($N_{\text{union}}$) within each $\uset$ construct, and the total number of distinct symbolic amplitudes ($N_{\text{amp}}$) in the assertion.
% \vspace{-5pt}
\begin{restatable}{theorem}{finalSpaceComplexity}\label{thm:finalSpaceComplexity}
The size of the final LSTA $\mathcal{M}$ constructed from an assertion is bounded by $
    |\Delta| = O\left(2^{N_{\text{term}} \cdot 2^{N_{\text{vc}}}} \cdot 2^{N_{\text{vc}}} \cdot N_{\text{term}} \cdot N_{\text{union}} \cdot L \cdot N_{\text{amp}}\right).
$
\end{restatable}%\vspace{-5pt}

Our complexity is \emph{linear} in the key parameter $L$. This efficiency stems from transforming assertions into tensor product structures, leveraging the additive complexity of LSTA tensor product operations with a scaling multiplier. The remaining structural parameters are independent of $L$ and negligible in practice; in our experiments, all such parameters are bounded by~$2$ (see \Cref{table:prepost} in \Cref{app:exp}).

% benchmarks
\newcommand{\bvsingbench}[0]{\textsc{BV-Sing}\normalsize\xspace}
\newcommand{\bvmultbench}[0]{\textsc{BV}\normalsize\xspace}
\newcommand{\ghzsingbench}[0]{\textsc{GHZ-Sing}\normalsize\xspace}
\newcommand{\ghzmultbench}[0]{\textsc{GHZ}\normalsize\xspace}
\newcommand{\groversingbench}[0]{\textsc{Grover-Sing}\normalsize\xspace}
\newcommand{\grovermultbench}[0]{\textsc{Grover}\normalsize\xspace}
\newcommand{\oegroverbench}[0]{\textsc{GroverIter}\normalsize\xspace}
\newcommand{\mctoffolibench}[0]{\textsc{MCToffoli}\normalsize\xspace}

\section{Experimental Results}\label{sec:Experimental Results}
We implemented the specification-to-automata translation framework proposed in this work and compared it with the translation algorithm in the latest version of \autoqqq~\cite{DBLP:conf/tacas/ChenCHHLLT25}\footnote{https://github.com/fmlab-iis/AutoQ}. The latter uses the same specification language as described in~\cite{DBLP:conf/cav/ChenCLLT23}, but translates specifications into LSTAs.

We evaluated performance on a suite of representative circuits: Bernstein-Vazirani (\bvmultbench), Greenberger-Horne-Zeilinger state preparation~\cite{Greenberger1989} (\ghzmultbench), Grover's search (\grovermultbench), one iteration of Grover's search (\oegroverbench), and the multi-controlled Toffoli gate (\mctoffolibench), as detailed in \Cref{sec:usecases}. These benchmarks encompass canonical quantum algorithms and standard composite gates. The preconditions and postconditions for these benchmarks are listed in \Cref{table:prepost} in \Cref{app:exp}.

For the oracle-based algorithms (\bvmultbench, \grovermultbench, and \oegroverbench), we verified their parameterized-oracle versions to guarantee correctness for all possible oracles. This was achieved by using the first $n$ qubits to control the X gates in the oracle via CX gates. For \ghzmultbench, we input all basis states to demonstrate the bitwise complement feature of our language. For \mctoffolibench, we verified the four specific cases indicated in \Cref{sec:multi-control} individually to ensure the functional correctness of the implementation.

We conducted all experiments on a server running Ubuntu 24.04.3 LTS, equipped with an AMD EPYC 7742 64-core processor (1.5 GHz), 2 TiB of RAM, and a 4 TB SSD. A timeout of 5 minutes was enforced for each circuit verification. The results are presented in \Cref{table:exp}. The proposed translation algorithm is consistently faster than the original \autoqq algorithm. Moreover, it produces smaller LSTAs, leading to shorter verification times. These results highlight the scalability achieved by our approach.\vspace{-10pt}

\begin{table}[htbp]
% tool names
\newcommand{\tool}[0]{\textsc{LSTAQ}\xspace}
\newcommand{\correct}[0]{T\xspace}
\newcommand{\wrong}[0]{F\xspace}
\newcommand{\unknown}[0]{---\xspace}
\newcommand{\nacell}[0]{\cellcolor{black!20}}
\newcommand{\TO}[0]{\nacell TIMEOUT}
\newcommand{\error}[0]{\nacell OOM}
    % 1. Caption 維持原樣，包含完整的欄位定義
    \caption{Results of verifying our use cases with this work and \cite{DBLP:conf/tacas/ChenCHHLLT25}. Columns \textbf{\#q} and \textbf{\#G} denote the number of qubits and gates of the circuit, respectively. For each case, we report the time required to: translate the specification into precondition and postcondition LSTAs (trans), perform the verification process (ver), and the total running time (total). The timeout is 5 minutes.}\vspace{-5pt}
    \label{table:exp}
    \centering
    \scriptsize
    \setlength{\tabcolsep}{2.8pt}
    
    % =============================================
    % [左半部] 前三組 Benchmark (共約 15 行數據)
    % =============================================
    \begin{minipage}[t]{0.49\linewidth}
        \centering
        \begin{tabular}{c c c ccc ccc}
        \toprule
        \multirow{2}{*}{} & \multirow{2}{*}{\textbf{\#q}} & \multirow{2}{*}{\textbf{\#G}} & \multicolumn{3}{c}{This Work} & \multicolumn{3}{c}{\autoqqq~\cite{DBLP:conf/tacas/ChenCHHLLT25}} \\
        \cmidrule(lr){4-6} \cmidrule(lr){7-9}
          & & & trans & ver & total & trans & ver & total \\
        \midrule
        % Group 1: BV
          \multirow{5}{*}{\rotatebox[origin=c]{90}{\bvmultbench}}
         & 17 & 27 & 0.0s & 0.0s & \textbf{0.0s} & 1.5s & 0.9s & \textbf{2.4s}\\
         & 19 & 30 & 0.0s & 0.0s & \textbf{0.0s} & 6.1s & 4.3s & \textbf{10.4s}\\
         & 21 & 33 & 0.0s & 0.0s & \textbf{0.0s} & 25s & 19.2s & \textbf{44.2s}\\
         & 23 & 36 & 0.0s & 0.0s & \textbf{0.0s} & 2m19s & 1m35s & \textbf{3m54s}\\
         & 25 & 39 & 0.0s & 0.0s & \textbf{0.0s} & \multicolumn{3}{c}{\TO}\\
        \midrule
        % Group 2: GHZ
          \multirow{5}{*}{\rotatebox[origin=c]{90}{\ghzmultbench}}
         & 8 & 8 & 0.0s & 0.0s & \textbf{0.0s} & 0.6s & 0.2s & \textbf{0.8s}\\
         & 9 & 9 & 0.0s & 0.0s & \textbf{0.0s} & 2.7s & 1.1s & \textbf{3.8s}\\
         & 10 & 10 & 0.0s & 0.0s & \textbf{0.0s} & 11s & 4.8s & \textbf{15.8s}\\
         & 11 & 11 & 0.0s & 0.0s & \textbf{0.0s} & 44s & 19.1s & \textbf{1m3s}\\
         & 12 & 12 & 0.0s & 0.0s & \textbf{0.0s} & \multicolumn{3}{c}{\TO}\\
        \midrule
        % Group 3: Grover
          \multirow{5}{*}{\rotatebox[origin=c]{90}{\grovermultbench}}
         & 20 & 544 & 0.0s & 0.2s & \textbf{0.2s} & 0.5s & 2.9s & \textbf{3.4s}\\
         & 23 & 927 & 0.0s & 0.5s & \textbf{0.5s} & 1.9s & 10.1s & \textbf{12s}\\
         & 26 & 1475 & 0.0s & 1s & \textbf{1s} & 7.3s & 42s & \textbf{49.3s}\\
         & 29 & 2408 & 0.0s & 1.9s & \textbf{1.9s} & 30s & 3m33s & \textbf{4m3s}\\
         & 32 & 3711 & 0.0s & 3.5s & \textbf{3.5s} & \multicolumn{3}{c}{\TO}\\
        \bottomrule
        \end{tabular}
    \end{minipage}
    % \hfill
    % =============================================
    % [右半部] 後兩組 Benchmark + 空白填充 + Footnote
    % =============================================
    \begin{minipage}[t]{0.49\linewidth}
        \centering
        \begin{tabular}{c c c ccc ccc}
        \toprule
        \multirow{2}{*}{} & \multirow{2}{*}{\textbf{\#q}} & \multirow{2}{*}{\textbf{\#G}} & \multicolumn{3}{c}{This Work} & \multicolumn{3}{c}{\autoqqq~\cite{DBLP:conf/tacas/ChenCHHLLT25}} \\
        \cmidrule(lr){4-6} \cmidrule(lr){7-9}
          & & & trans & ver & total & trans & ver & total \\
        \midrule
        % Group 4: OE Grover
          \multirow{5}{*}{\rotatebox[origin=c]{90}{\oegroverbench}}
         & 20 & 79 & 0.0s & 0.1s & \textbf{0.1s} & 0.8s & 2.7s & \textbf{3.5s}\\
         & 23 & 91 & 0.0s & 0.1s & \textbf{0.1s} & 2.9s & 8.3s & \textbf{11.2s}\\
         & 26 & 103 & 0.0s & 0.1s & \textbf{0.1s} & 11s & 30.7s & \textbf{41.7s}\\
         & 29 & 115 & 0.0s & 0.1s & \textbf{0.1s} & 50s & 2m19s & \textbf{3m9s}\\
         & 32 & 127 & 0.0s & 0.1s & \textbf{0.1s} & \multicolumn{3}{c}{\TO}\\
        \midrule
        % Group 5: MC Toffoli
          \multirow{5}{*}{\rotatebox[origin=c]{90}{\mctoffolibench\scriptsize\hspace{-2pt}\textsuperscript{*}}} % 手動標記 a
         & 16 & 15 & 0.0s & 0.0s & \textbf{0.0s} & 1.2s & 0.4s & \textbf{1.6s}\\
         & 18 & 17 & 0.0s & 0.0s & \textbf{0.0s} & 4.4s & 1.6s & \textbf{6s}\\
         & 20 & 19 & 0.0s & 0.0s & \textbf{0.0s} & 18.1s & 7.3s & \textbf{25.4s}\\
         & 22 & 21 & 0.0s & 0.0s & \textbf{0.0s} & 1m13s & 30.4s & \textbf{1m43s}\\
         & 24 & 23 & 0.0s & 0.0s & \textbf{0.0s} & \multicolumn{3}{c}{\TO}\\
        \midrule
         \multicolumn{9}{l}{*) Running time in this benchmark is the sum of four cases}\\
         \multicolumn{9}{l}{indicated in \Cref{sec:multi-control}.}\\\\\\\\\arrayrulecolor{white}\bottomrule\arrayrulecolor{black}
        \end{tabular}
    \end{minipage}
\end{table}

\paragraph{Cross-Paradigm Comparison.} 
Beyond the performance improvements in translation complexity, the tree-automata-based paradigm utilized in this work has been extensively benchmarked against various automatic verification tools in prior studies~\cite{DBLP:journals/pacmpl/AbdullaCCHLLLT25,DBLP:conf/tacas/ChenCHHLLT25,DBLP:conf/cav/ChenCLLT23,pldi23,cacm25}. These include symbolic verifiers such as \textsc{symQV}~\cite{DBLP:conf/fm/BauerMarquartLS23} (based on the SMT theory of reals) and \textsc{CaAL}~\cite{DBLP:conf/cade/ChenRT23} (based on an extended SMT theory of arrays), as well as simulators like the state-vector-based \textsc{SV-Sim}~\cite{DBLP:conf/sc/LiFGPHRK21} and the decision-diagram-based \textsc{SliQSim}~\cite{DBLP:conf/dac/TsaiJJ21}. Furthermore, the paradigm has been compared with the \textsc{Feynman} tool suite~\cite{DBLP:journals/corr/abs-1805-06908} (based on the path-sum) and \textsc{Qcec}~\cite{DBLP:journals/tcad/BurgholzerW21} (which integrates decision diagrams, the ZX-calculus~\cite{Coecke2011}, and random stimuli generation~\cite{DBLP:conf/aspdac/BurgholzerKW21}). These evaluations demonstrate that while primarily focusing on pure-state families, the set-based, automata-driven approach outperforms these counterparts in terms of running time and scalability, particularly for large-scale instances. We omit direct comparisons with Hermitian-based or projection-based tools as they typically prioritize the breadth of mixed-state logic and often require non-trivial manual proof developments in interactive theorem provers, whereas our approach provides a fully automated, ``push-button'' verification path for large-scale circuits.

%%%%%%%%%%%%%%%%%%%%%%%%%%%%%%%%%%%%%%%%%%%%%%%%%%%%%%%%%%%
\section{Concluding Remarks}\label{sec:Conclusion}

We have bridged the gap between expressive high-level specifications and fully automated quantum program verification. By identifying the exponential blow-up in prior automata-based approaches as a primary bottleneck, we introduced an extended specification language capable of describing complex families of quantum states and proposed a novel translation algorithm. By leveraging variable-level and qubit-level reordering strategies to facilitate tensor product decomposition, our algorithm reduces the construction complexity of LSTAs from exponential to linear in the number of qubits. Our experimental results demonstrate that this approach dramatically improves scalability, enabling the verification of large-scale circuits that were previously beyond the reach of existing automated tools.

To better situate our framework within the landscape of quantum verification, it is important to distinguish its underlying philosophy from traditional methods. Our approach adopts an extensional view of specifications, where predicates are identified with sets of quantum states and represented using a compact automata-based structure. Although this differs from the intensional, formula-based approach in Hoare-style logics, it still supports principled reasoning through set-theoretic operations such as inclusion, union, and intersection, reflecting a different balance between expressiveness and tractability. An important property of LSTAs is that they are not closed under complement. Supporting arbitrary negation would require leaving the automata domain, thereby sacrificing compactness and automation. We therefore treat this as a deliberate restriction to preserve tractability. Such restrictions are not entirely foreign to logical systems. For instance, in intuitionistic logic, negation is not treated as a primitive operator in the same way as in classical logic, and reasoning often proceeds via implication. Similarly, while complement is not supported in our setting, we retain a meaningful notion of implication through set inclusion, which can be efficiently checked within our automata framework. Overall, we view our approach as complementary to QHL-style frameworks: logical specifications provide generality and strong reasoning principles, while our representation-aligned design enables scalable and fully automatic verification in a practically important regime.

Looking ahead, there are several promising avenues for extending this framework. One potential direction involves the automatic splitting of atomic variables to facilitate variable alignment, which is particularly beneficial when manual decomposition is complex. Such an extension could be realized by enhancing $\varcon$ constructs to support logical disjunction. Additionally, the framework could be further extended to support classical-quantum states, enabling the verification of programs with mid-circuit measurements and sophisticated classical control flow. This might be achieved by, for instance, employing equalities to bind classical and quantum states to variables, offering a more flexible representation. These developments would further enhance the practical utility and automated capabilities of our framework, ultimately paving the way for verifying realistic hybrid quantum-classical workflows.

% --- Data-Availability Statement ---
\subsubsection*{Data-Availability Statement.}
The code and benchmarks for reproducing the findings of this paper are available on Zenodo at \url{https://doi.org/10.5281/zenodo.19724316}~\cite{artifact}.

\subsubsection*{Acknowledgments.}
We thank the anonymous reviewers for their constructive feedback. This work was supported by National Science and Technology Council, R.O.C., project NSTC 114-2119-M-001-002; Air Force Office of Scientific Research project FA2386-23-1-4107; Academia Sinica Investigator Project Grant AS-IV-114-M07; the Czech Science Foundation project 25-18318S; and the FIT BUT internal project FIT-S-26-9011.

% --- Disclosure of Interests ---
\subsubsection*{Disclosure of Interests.}
The authors have no relevant financial or non-financial interests to disclose.

%%%%%%%%%%%%%%%%%%%%%%%%%%%%%%%%%%%%%%%%%%%%%%%%%%%%%%%%%%%%%
% \clearpage
\bibliographystyle{splncs04}
\bibliography{literature}
%%%%%%%%%%%%%%%%%%%%%%%%%%%%%%%%%%%%%%%%%%%%%%%%%%%%%%%%%%%%%

% \newpage
\appendix
\section{Appendix}

\subsection{Binary Operations of LSTA}\label{app:bin}

\begin{lemma}[Semantics of LSTA Set Union]\label{lem:unionSemantics}
Given two LSTAs $\mathcal{A}$ and $\mathcal{B}$ over $\mathbb{K}$ where $\mathcal L(\mathcal A)$ and $\mathcal L(\mathcal B)$ contain $n$-qubit and $m$-qubit states, respectively, there exists a set union operation, denoted by $\mathcal{A} \sqcup \mathcal{B}$, which yields a valid LSTA over $\mathbb{K}$. The recognized language of the resulting LSTA satisfies:
\[
    \mathcal{L}(\mathcal{A} \sqcup \mathcal{B}) = \mathcal{L}(\mathcal{A}) \cup \mathcal{L}(\mathcal{B}).
\]
\end{lemma}
\begin{proof}
Given two LSTAs $\mathcal{A} = \langle Q_\mathcal{A}, V_\mathcal{A}, \Delta_\mathcal{A}, r_\mathcal{A}\rangle$ and $\mathcal{B} = \langle Q_\mathcal{B}, V_\mathcal{B}, \Delta_\mathcal{B}, r_\mathcal{B}\rangle$, we compute their set union $\mathcal{A} \sqcup \mathcal{B}$ as follows such that $\mathcal L(\mathcal{A} \sqcup \mathcal{B}) = \mathcal L(\mathcal{A}) \cup \mathcal L(\mathcal{B})$. Without loss of generality, we assume $Q_\mathcal{A} \cap Q_\mathcal{B} = \emptyset$. Intuitively, the set union automaton acts as a selector that chooses to enter either $\mathcal{A}$ or $\mathcal{B}$ based on the choice of the root transition. This can be constructed by merging the two automata and modifying the root transitions. We introduce a new root state $r_\cup$ and define the new transition set as follows. Let $\{\delta_1, \delta_2, \dots, \delta_k\} = \Delta^\mathcal{A}_r \cup \Delta^\mathcal{B}_r$ be an arbitrary enumeration of all original root transitions. We construct a set of transitions originating from $r_\cup$ with sequentially re-indexed choices $\Delta^\cup_r = \{ r_\cup \xrightarrow{\{i\}} (\leftof{\delta_i},\ \rightof{\delta_i}) \mid 1 \le i \le k \}$. The resulting $\mathcal{A} \sqcup \mathcal{B} \triangleq \langle Q_\mathcal{A} \cup Q_\mathcal{B} \cup \{r_\cup\}, V_\mathcal{A} \cup V_\mathcal{B}, (\Delta_\mathcal A \setminus \Delta^\mathcal{A}_r) \cup (\Delta_\mathcal B \setminus \Delta^\mathcal{B}_r) \cup \Delta^\cup_r, r_\cup\rangle$. The sequential choice re-indexing $\{1\}, \dots, \{k\}$ trivially ensures the choice disjointness condition.
\qed
\end{proof}

\begin{lemma}[Semantics of LSTA Tensor Product]\label{lem:tensorSemantics}
Given two LSTAs $\mathcal{A}$ and $\mathcal{B}$ over $\mathbb{K}$ where $\mathcal L(\mathcal A)$ and $\mathcal L(\mathcal B)$ contain $n$-qubit and $m$-qubit states, respectively, there exists a tensor product operation, denoted by $\mathcal{A} \otimes \mathcal{B}$, which yields a valid LSTA over $\mathbb{K}$. Furthermore, the recognized language of the resulting LSTA satisfies:
\[
    \mathcal{L}(\mathcal{A} \otimes \mathcal{B}) = \mathcal{L}(\mathcal{A}) \otimes \mathcal{L}(\mathcal{B}).
\]
\end{lemma}
\begin{proof}
Given two LSTAs $\mathcal{A} = \langle Q_\mathcal{A}, V_\mathcal{A}, \Delta_\mathcal{A}, r_\mathcal{A}\rangle$ and $\mathcal{B} = \langle Q_\mathcal{B}, V_\mathcal{B}, \Delta_\mathcal{B}, r_\mathcal{B}\rangle$, we build their tensor product $\mathcal{A} \otimes \mathcal{B}$ such that $\mathcal L(\mathcal{A} \otimes \mathcal{B}) = \mathcal L(\mathcal{A}) \otimes \mathcal L(\mathcal{B})$. Recall that, in the tree view, the tensor product $\ket{\psi} \otimes \ket{\phi}$ essentially replaces each scalar amplitude $a_s$ in $\ket{\psi}$ with the subtree representing $a_s\ket{\phi}$. Thus, the intuition behind this construction is to structurally mimic the algebraic expansion of the tensor product. That is, we replace all leaves of a tree induced by $\mathcal{A}$ with an identical tree induced by $\mathcal{B}$, scaled by the corresponding leaf values. Formally, the construction proceeds in two steps:
\begin{enumerate}[(1)]
    \item \textbf{Replication and Scaling:} For each distinct amplitude $\alpha \in \{ \symof{\delta} \mid \delta \in \Delta^{\mathcal A}_\textit{ex} \}$, we create a distinct, isomorphic scaled copy of $\mathcal{B}$, denoted by $\mathcal{B}_\alpha$.
    In this copy, every state $q \in Q_\mathcal{B}$ is renamed to $q^\alpha$, and every leaf amplitude $v$ is scaled to $\alpha \cdot v$. Specifically, the transition set $\Delta_{\mathcal{B}_\alpha}$ corresponds one-to-one with $\Delta_\mathcal{B}$: internal transitions are structurally identical (from $q \xrightarrow{C} (q_\ell, q_r)$ to $q^\alpha \xrightarrow{C} (q_\ell^\alpha, q_r^\alpha)$), while leaf transitions are scaled (from $q \xrightarrow{C} v$ to $q^\alpha \xrightarrow{C} \alpha \cdot v$).
    \item \textbf{Interface Construction:} The interface transitions $\Delta_{\textit{join}}$ fusing $\mathcal{A}$'s leaves and $\mathcal{B}_\alpha$'s roots are constructed as follows. Let $U^\mathcal{A}_\textit{ex} = \bigcup_{\delta \in \Delta^\mathcal{A}_\textit{ex}} \chof{\delta}$, $U^\mathcal{B}_\textit{r} = \bigcup_{\delta \in \Delta^\mathcal{B}_\textit{r}} \chof{\delta}$, and $U^\mathcal{A}_\textit{in} = \bigcup_{\delta \in \Delta^\mathcal{A}_\textit{in}} \chof{\delta}$ be the choice universes of $\Delta^\mathcal{A}_\textit{ex}$, $\Delta^\mathcal{B}_\textit{r}$, and $\Delta^\mathcal{A}_\textit{in}$, respectively. We fix an injection $f: U^\mathcal{A}_\textit{ex} \times U^\mathcal{B}_\textit{r} \to \mathbb{N} \setminus U^\mathcal{A}_\textit{in}$. Then, for every leaf transition $q \xrightarrow{C_1} \alpha$ of $\mathcal{A}$ and every root transition $r_\mathcal{B} \xrightarrow{C_2} (q_\ell, q_r)$ of $\mathcal{B}$, we add a transition $q \xrightarrow{\{f(a, b)\ \mid\ a \in C_1,\ b \in C_2\}} (q_\ell^\alpha, q_r^\alpha)$ into $\Delta_{\textit{join}}$.
\end{enumerate}
The resulting product LSTA is then defined as $\mathcal{A} \otimes \mathcal{B} \triangleq \langle Q_\mathcal{A} \cup \bigcup_{\alpha} Q_{\mathcal{B}_\alpha}, V_\mathcal A \cup V_\mathcal B, \Delta^\mathcal{A}_\textit{in} \cup \Delta_{\textit{join}} \cup \bigcup_{\alpha} (\Delta_{\mathcal{B}_\alpha} \setminus \Delta^{\mathcal{B}_\alpha}_r), r_\mathcal{A}\rangle$. The injectivity of $f$ ensures that this construction satisfies the choice disjointness condition, as explained in detail below.

\paragraph{Satisfaction of Choice Disjointness Condition.}
To prove that $\mathcal{A} \otimes \mathcal{B}$ is a valid LSTA, we must verify that for any state $q$ in the product and any two distinct transitions $\delta_1, \delta_2$ originating from $q$, their choice sets are disjoint, i.e., $\chof{\delta_1} \cap \chof{\delta_2} = \emptyset$.

First, consider the states that are not affected by the interface construction. For any state $q$ belonging to the scaled copies $\bigcup_{\alpha} Q_{\mathcal{B}_\alpha}$, the transitions originate solely from $\Delta_{\mathcal{B}_\alpha}$. Since $\mathcal{B}$ is a valid LSTA and $\mathcal{B}_\alpha$ is an isomorphic copy, choice disjointness is preserved by isomorphism. Similarly, for any non-interface state $q \in Q_\mathcal{A}$ (i.e., a state in $\mathcal{A}$ that has no external transitions), the set of outgoing transitions remains a subset of $\Delta^\mathcal{A}_\textit{in}$. Thus, the condition holds by inheritance from $\mathcal{A}$.

The critical case arises for an \textbf{interface state} $q \in Q_\mathcal{A}$, which may now possess transitions from both the original $\Delta^\mathcal{A}_\textit{in}$ and the new $\Delta_\textit{join}$. Let $\delta_1, \delta_2$ be two distinct transitions with $\topof{\delta_1} = \topof{\delta_2} = q$. We analyze the two possible violating scenarios:

\begin{itemize}
    \item \textbf{Case 1: Join vs. Internal.} Suppose $\delta_1 \in \Delta_\textit{join}$ and $\delta_2 \in \Delta^\mathcal{A}_\textit{in}$.
    By the construction of the interface, the choice set $\chof{\delta_1}$ of $\delta_1$ is a subset of the codomain of $f$, which is defined as $\mathbb{N} \setminus U^\mathcal{A}_\textit{in}$. In contrast, $\delta_2$ is an original internal transition of $\mathcal{A}$, so its choice set $\chof{\delta_2}$ is a subset of $U^\mathcal{A}_\textit{in}$.
    Since $(\mathbb{N} \setminus U^\mathcal{A}_\textit{in}) \cap U^\mathcal{A}_\textit{in} = \emptyset$, it follows that $\chof{\delta_1} \cap \chof{\delta_2} = \emptyset$.

    \item \textbf{Case 2: Join vs. Join.} Suppose $\delta_1, \delta_2 \in \Delta_\textit{join}$ are both new interface transitions.
    By definition, these transitions are derived from pairs of original transitions. Let $\delta_1$ be derived from the pair $(\delta^\mathcal{A}_1, \delta^\mathcal{B}_1)$ where $\delta^\mathcal{A}_1 = q \xrightarrow{C_1} \alpha$ and $\delta^\mathcal{B}_1 = r_\mathcal{B} \xrightarrow{D_1} (q_{\ell,1}, q_{r,1})$. Similarly, let $\delta_2$ be derived from $(\delta^\mathcal{A}_2, \delta^\mathcal{B}_2)$ with choices $C_2, D_2$ and children $(q_{\ell,2}, q_{r,2})$.
    The choice sets for the new transitions are $f(C_1 \times D_1)$ and $f(C_2 \times D_2)$.
    Since $f$ is injective, the intersection of their image sets is:
    \begin{align*}
        f(C_1 \times D_1) \cap f(C_2 \times D_2) &= f\big( (C_1 \times D_1) \cap (C_2 \times D_2) \big) \\
        &= f\big( (C_1 \cap C_2) \times (D_1 \cap D_2) \big).
    \end{align*}
    Since $\delta_1 \neq \delta_2$, at least one component of the source pairs must differ:
    \begin{enumerate}
        \item If $\delta^\mathcal{A}_1 \neq \delta^\mathcal{A}_2$, then by the validity of $\mathcal{A}$, we have $C_1 \cap C_2 = \emptyset$.
        \item If $\delta^\mathcal{A}_1 = \delta^\mathcal{A}_2$, then we must have $\delta^\mathcal{B}_1 \neq \delta^\mathcal{B}_2$ (otherwise $\delta_1 = \delta_2$). By the validity of $\mathcal{B}$, we have $D_1 \cap D_2 = \emptyset$.
    \end{enumerate}
    In either subcase, the Cartesian product $(C_1 \cap C_2) \times (D_1 \cap D_2)$ is empty, and hence $\chof{\delta_1} \cap \chof{\delta_2} = \emptyset$.
\end{itemize}
Since disjointness holds for all cases, $\mathcal{A} \otimes \mathcal{B}$ satisfies the choice disjointness condition.

\paragraph{Proof of $\mathcal{L}(\mathcal{A} \otimes \mathcal{B}) = \mathcal{L}(\mathcal{A}) \otimes \mathcal{L}(\mathcal{B})$.}
We prove the equality by showing mutual inclusion based on the existence of valid choice sequences satisfying the LSTA conditions.

\textbf{\noindent(Proof of $\supseteq$):}
Let $\ket{\psi} = \sum_{s} a_s \ket{s} \in \mathcal{L}(\mathcal{A})$ and $\ket{\phi} = \sum_{t} b_t \ket{t} \in \mathcal{L}(\mathcal{B})$. We aim to show that the product state $\ket{\Psi} = \ket{\psi} \otimes \ket{\phi}$ belongs to $\mathcal{L}(\mathcal{A} \otimes \mathcal{B})$.
By definition, there exists a choice sequence $c^{\psi}_1, \dots, c^{\psi}_n, c^{\psi}_0$ for $\mathcal{A}$ that induces a tree representing $\ket{\psi}$, and a sequence $c^{\phi}_1, \dots, c^{\phi}_m, c^{\phi}_0$ for $\mathcal{B}$ representing $\ket{\phi}$.

We construct the choice sequence $c^{\Psi}_1, \dots, c^{\Psi}_{n+m}, c^{\Psi}_0$ for $\ket{\Psi}$ as follows:
\begin{itemize}
    \item \textbf{Upper Levels ($1 \le i \le n$):} Set $c^{\Psi}_i = c^{\psi}_i$. These choices guide the path through $\Delta^\mathcal{A}_\textit{in}$ exactly as in $\ket{\psi}$.
    \item \textbf{Interface Level ($i = n+1$):} Set $c^{\Psi}_{n+1} = f(c^{\psi}_0, c^{\phi}_1)$. Since $f$ is an injection, this choice uniquely identifies the transition in $\Delta_\textit{join}$ derived from the leaf transition of $\ket{\psi}$ (choice $c^{\psi}_0$) and the root transition of $\ket{\phi}$ (choice $c^{\phi}_1$).
    \item \textbf{Lower Levels ($n+2 \le i \le n+m$):} Set $c^{\Psi}_i = c^{\phi}_{i-n}$. These choices guide the path through the copies $\mathcal{B}_\alpha$.
    \item \textbf{Leaves ($i = 0$):} Set $c^{\Psi}_0 = c^{\phi}_0$.
\end{itemize}
We now verify the validity of the induced tree for any basis state $\ket{st}$ where $s \in \{0,1\}^n$ and $t \in \{0,1\}^m$.
The sequence $c^{\Psi}_1, \dots, c^{\Psi}_n$ guides the automaton from $r_\mathcal{A}$ to a state $q$ (the leaf of path $s$ in $\mathcal{A}$), which originally had a leaf transition to value $a_s$ via $c^{\psi}_0$.
At level $n+1$, the choice $c^{\Psi}_{n+1}$ selects the unique interface transition, which enters the copy of $\mathcal{B}$ (scaled by $a_s$) corresponding to the first level of path $t$.
The subsequent choices $c^{\Psi}_{n+2}, \dots, c^{\Psi}_{n+m}, c^{\Psi}_0$ then guide the automaton through $\mathcal{B}_{a_s}$ following path $t$ starting from the second level. Since $\mathcal{B}_{a_s}$ is an isomorphic copy of $\mathcal{B}$ scaled by $a_s$, the final leaf value reached is $a_s \cdot b_t$.
Thus, the generated state matches $\sum_{s,t} (a_s \cdot b_t) \ket{st} = \ket{\psi} \otimes \ket{\phi} = \ket{\Psi}$.

\textbf{\noindent(Proof of $\subseteq$):}
Let $\ket{\Psi} \in \mathcal{L}(\mathcal{A} \otimes \mathcal{B})$. By definition, there exists a choice sequence $c^{\Psi}_1, \dots, c^{\Psi}_{\ell}, c^{\Psi}_0$ that induces a valid tree representing $\ket{\Psi}$.

\textbf{First, we verify that any perfect binary tree induced by the transitions of $\mathcal{A} \otimes \mathcal{B}$ must consist of exactly $n+m$ internal levels, thereby representing an $(n+m)$-qubit state (i.e., $\ell = n+m$).} This structure arises because the leaf transitions of the product automaton belong exclusively to the scaled copies $\{\mathcal{B}_\alpha\}$. To reach these leaves, any valid path starting from $r_\mathcal{A}$ must traverse the structure of $\mathcal{A}$ and enter a copy $\mathcal{B}_\alpha$ via the interface $\Delta_\textit{join}$, so at least the upper part of the tree representing $\ket{\Psi}$---which is also a perfect binary tree---is induced solely by transitions in $\Delta^\mathcal{A}_\textit{in}$.

Suppose the selected transitions in $\Delta_\textit{join}$ are enabled by choice $c^{\Psi}_{k}$ at level $k$. Then $c^{\Psi}_{k}$ must be representable as $f(c_\mathcal{A}, c_\mathcal{B})$ for some $c_\mathcal{A} \in U^\mathcal{A}_\textit{ex}$ and $c_\mathcal{B} \in U^\mathcal{B}_\textit{r}$.
Consider the upper part: by replacing the selected interface transitions with the original leaf transitions of $\mathcal{A}$ (enabled by $c_\mathcal{A}$) and combining them with the enabled transitions in $\Delta^{\mathcal A}_\textit{in}$, we recover a perfect binary tree representing a quantum state in $\mathcal{L}(\mathcal{A})$. Since every quantum state in $\mathcal{L}(\mathcal{A})$ has $n$ qubits, the interface must occur at level $k = n+1$.

Similarly, regarding the lower part of the tree, if we replace the selected interface transitions with the original root transitions in $\Delta^{\mathcal{B}}_r$ (enabled by $c_\mathcal{B}$) and combine them with the transitions in $\Delta_\mathcal{B}$ recovered from the enabled transitions in $\{\Delta_{\mathcal{B}_\alpha}\}$, we recover a perfect binary tree representing a quantum state in $\mathcal{L}(\mathcal{B})$. Since every quantum state in $\mathcal{L}(\mathcal{B})$ has $m$ qubits, this lower part contributes exactly $m$ internal levels. Thus, we obtain $\ell = n+m$.

Given $\ell = n+m$, we can safely decompose the sequence $c^{\Psi}_1, \dots, c^{\Psi}_{n+m}, c^{\Psi}_0$ into two parts along with $c_\mathcal A$ and $c_\mathcal B$ introduced previously:
\begin{itemize}
    \item \textbf{Upper Structure:} The sequence $c^{\Psi}_1, \dots, c^{\Psi}_n, c_\mathcal A$ forms a valid choice sequence for $\mathcal{A}$. The choices $c^{\Psi}_1 \dots c^{\Psi}_n$ guide the path through the $n$ internal levels of $\mathcal{A}$. The existence of the interface transitions at level $n+1$, enabled by $f(c_\mathcal A, c_\mathcal B)$ implies that for every path reaching this level, there was originally a leaf transition in $\mathcal{A}$ enabled by $c_\mathcal A$. Let $\ket{\psi} = \sum a_s \ket{s} \in \mathcal L(\mathcal A)$ be the quantum state induced by this choice sequence.
    \item \textbf{Lower Structure:} The sequence $c_\mathcal B, c^{\Psi}_{n+2}, \dots, c^{\Psi}_{n+m}, c^{\Psi}_0$ forms a valid choice sequence for $\mathcal{B}$. Since all copies $\mathcal{B}_\alpha$ are isomorphic, these choices induce the same structural path for the remaining $m$ levels for every branch $s$. Let $\ket{\phi} = \sum b_t \ket{t} \in \mathcal L(\mathcal B)$ be the quantum state induced by this choice sequence.
\end{itemize}
In the product tree, every branch $s$ (reaching amplitude $a_s$ in $\mathcal{A}$) will definitely enter the scaled copy $\mathcal{B}_{a_s}$. The lower choices then traverse $\mathcal{B}_{a_s}$ to reach the leaf for $t$. Due to the scaling factor $a_s$, the value at leaf $st$ is $a_s \cdot b_t$.
Therefore, $\ket{\Psi} = \sum_{s,t} (a_s b_t) \ket{st} = \ket{\psi} \otimes \ket{\phi} \in \mathcal L(\mathcal A) \otimes \mathcal L(\mathcal B)$.
\qed
\end{proof}

\binaryOperationComplexity* % 加星號 * 代表不要重新編號
\hypertarget{pfBinaryOperationComplexity}{
\begin{proof}
The proof follows from \Cref{lem:unionSemantics} and \Cref{lem:tensorSemantics}.
\qed
\end{proof}}

\subsection{Space Complexity of LSTA}\label{app:space}
We first give the proof of \Cref{thm:atomicSpaceComplexity}, stating the space complexity for exactly one quantum state, which is the basic building block for deriving the final space complexity.

\atomicSpaceComplexity*
\hypertarget{pfAtomicSpaceComplexity}{\begin{proof}
The LSTA construction consists of $n$ internal levels and one final leaf level, for a total of $n+1$ levels. At each level (whether internal or leaf), the number of transitions is bounded by the number of nonzero-amplitude paths, which is at most $N$, plus one sink transition for zero-amplitude paths.
Thus, the total number of transitions is bounded by $(N+1)(n+1)$.
Expanding this term yields $Nn + N + n + 1$, and since $N \cdot n$ is the dominant term, the asymptotic size complexity is $|\Delta| = O(N \cdot n)$.
\qed
\end{proof}}

We then give the proof of \Cref{thm:finalSpaceComplexity}, stating the final construction complexity.

\finalSpaceComplexity*
\hypertarget{pfFinalSpaceComplexity}{}
\begin{proof}
The size analysis proceeds by examining the construction of individual $\setq$ constructs (specifically, atomic components containing a single Dirac notation) and their subsequent assembly via set union and tensor product operations.

\paragraph{Complexity of Composition Operations.}
Recall the bounds established in \Cref{theorem:Binary Operations}:
\begin{itemize}
    \item \hypertarget{set union}{\textbf{Set Union ($\cup$):}} The union operation merges transition sets and adjusts root transitions. It is strictly additive: if $\mathcal{A} = \mathcal{A}_1 \cup \mathcal{A}_2$, then $|\mathcal{A}| \le |\mathcal{A}_1| + |\mathcal{A}_2|$. This operation introduces no structural blow-up.
    \item \hypertarget{tensor product}{\textbf{Tensor Product ($\otimes$):}} The tensor product $\mathcal{A} \otimes \mathcal{B}$ grafts a scaled copy of $\mathcal{B}$ onto each distinct leaf of $\mathcal{A}$. Its size is bounded by $|\mathcal{A} \otimes \mathcal{B}| \le |\mathcal{A}| + N_{\text{leaves}}(\mathcal{A}) \cdot |\mathcal{B}|$, where $N_{\text{leaves}}(\mathcal{A})$ denotes the count of distinct amplitudes at the leaves of $\mathcal{A}$.
\end{itemize}

\paragraph{Variable-Level Transformation.}
Note that variable-level reordering is an optional step. While variable dependencies may restrict its application, successful reordering strictly reduces the effective constraint count $N_{\text{vc}}$ without altering the structural form of the complexity bound derived below. Our analysis assumes the worst-case scenario where no variable-level reordering is performed, ensuring that the derived bound remains valid regardless of whether this step is applied.

\paragraph{Qubit-Level Transformation.} 
Consider the construction of an \emph{atomic} $\setV$ component (corresponding to a single Dirac notation). Let $N_{\text{slot}}$ be the number of (variable-level) slots involved, and $L_\text{slot}$ be the number of qubits per slot. We analyze its transformation into a tensor product structure of $L_\text{slot}$ $\setQ$ constructs, built qubitwise. We introduce auxiliary parameters to distinguish constraint sources: let $N_v^S$ and $N_{\text{ineq}}^S$ denote the counts of iterating variables and inequality constraints in the set predicate, respectively. Similarly, let $N_v^T$ and $N_{\text{ineq}}^T$ denote those under the local summation within a term.

For the $k$-th qubit-level $\setQ$ construct ($1 \le k \le L_\text{slot}$), the construction involves:
\begin{enumerate}[(1)]
    \item $2^{N_v^S}$ quantum states derived from set predicate variables;
    \item An expansion of each state into $2^{N_v^T} \cdot N_{\text{term}}$ terms (accounting for local summations);
    \item Each state occupies $N_{\text{slot}}$ qubits.
\end{enumerate}
Consequently, the size of each base $\setV$ component $\mathcal{A}_k$ is bounded by $O(2^{N_v^S + N_v^T} \cdot N_{\text{term}} \cdot N_{\text{slot}})$ by \hyperlink{set union}{set union} ($2^{N_v^S}$) and \Cref{thm:atomicSpaceComplexity} ($2^{N_v^T} N_{\text{term}} \cdot N_{\text{slot}}$).

The final LSTA for this atomic component is the tensor product of these $L_\text{slot}$ parts: $\bigotimes_{k=1}^{L_\text{slot}} \mathcal{A}_k$.
During this composition, the number of \hyperlink{def:valuation-dependent amplitudes}{valuation-dependent amplitudes} is bounded by $W = 2^{(2^{N_{\text{ineq}}^S + N_{\text{ineq}}^T}) \cdot N_{\text{term}}}$.
Applying the \hyperlink{tensor product}{tensor product} bound iteratively, the total size is the sum of $\setQ$ sizes scaled by $W$:
\[
    |\mathcal{A}_{\text{atomic}}| = O\left( W \cdot \sum_{k=1}^{L_\text{slot}} |\mathcal{A}_k| \right) 
    = O\left( \underbrace{2^{(2^{N_{\text{ineq}}^S + N_{\text{ineq}}^T}) \cdot N_{\text{term}}}}_{W} \cdot L_\text{slot} \cdot \underbrace{2^{N_v^S + N_v^T} \cdot N_{\text{term}} \cdot N_{\text{slot}}}_{|\mathcal{A}_k|} \right).
\]
Substituting the size relation $L_V = N_{\text{slot}} \times L_\text{slot}$ and incorporating the valuation-dependent factor $W$, we obtain:
\[
    |\mathcal{A}_{\text{atomic}}| = O\left( \underbrace{2^{N_{\text{term}} \cdot 2^{N_{\text{ineq}}}}}_{W} \cdot 2^{N_v} \cdot N_{\text{term}} \cdot L_V \right),
\]
where $N_{\text{ineq}} = N_{\text{ineq}}^S + N_{\text{ineq}}^T$ and $N_v = N_v^S + N_v^T$.

\paragraph{Set Union Operations.}
We now proceed to resolve the remaining set union and tensor product operators in the assertion. We resolve set union operators first. Recall that in the \hyperlink{Canonicalization}{canonicalization} step (\Cref{sec:preprocessing}), a $\setq$ construct containing multiple $\dirac$ notations is expanded into a union of atomic $\setq$ constructs. The parameter $N_\text{union}$ ensures that each $\uset$ construct consists of at most $N_\text{union}$ such atomic components.
Since the LSTA \hyperlink{set union}{set union} operation is strictly additive, the size of a full $\uset$ construct is bounded by the sum of its $N_\text{union}$ atomic components. This introduces a linear scaling factor $2N_\text{union}$ (accounting for both separate constructions and the subsequent union operations), leading to the complexity:
\[
    |\mathcal{A}_{\uset}| = O\left( \underbrace{2^{N_{\text{term}} \cdot 2^{N_{\text{ineq}}}}}_{W} \cdot 2^{N_v} \cdot N_{\text{term}} \cdot N_\text{union} \cdot L_V \right).
\]

\paragraph{Tensor Product Operations.} 
The final phase resolves the remaining global tensor products. Since tensor product operations accumulate the number of qubits from operands ($L = \sum_V L_V$) and introduce the multiplier $N_{\text{amp}}$ (the number of symbolic amplitudes), the final complexity is:
\[
    |\mathcal{A}_{\text{final}}| = O\left( 2^{N_{\text{term}} \cdot 2^{N_{\text{ineq}}}} \cdot 2^{N_v} \cdot N_{\text{term}} \cdot N_\text{union} \cdot L \cdot N_{\text{amp}} \right).
\]

\paragraph{Overall Size Complexity.} 
To present a unified bound, we recall that $N_{\text{vc}}$ represents the cumulative constraint complexity, encompassing both $N_{\text{ineq}}$ and $N_v$ (i.e., $N_{\text{ineq}}, N_v \le N_{\text{vc}}$).
Substituting these into the derived bound, the final complexity simplifies to:
\[
    |\mathcal{A}_{\text{final}}| = O\left( 2^{N_{\text{term}} \cdot 2^{N_{\text{vc}}}} \cdot 2^{N_{\text{vc}}} \cdot N_{\text{term}} \cdot N_{\text{union}} \cdot L \cdot N_{\text{amp}} \right).
\]
This formulation confirms that the LSTA construction complexity remains \emph{linear} in the total number of qubits $L$, given that the other system parameters are fixed.
\qed
\end{proof}

\subsection{Experimental Settings}\label{app:exp}
We present the related experimental settings in this section, detailing the specific benchmarks and their corresponding formal specifications. The formal preconditions (initial states) and postconditions (expected output states) for verifying input-output correctness are systematically summarized below.

\begin{table}\caption{Experiment preconditions and postconditions. Notation $b^m$ denotes $m$ repetitions of bit $b$. For \mctoffolibench, ``flip'' (resp. ``$\neg$flip'') indicates that all control qubits are 1 (resp. not all 1), and $t$ denotes the initial target qubit value.}\label{table:prepost}
    \centering
    \scalebox{0.98}{
    \begin{tabular}{|c|c|c|}\hline
        \textbf{Benchmark} & \textbf{Precondition} & \textbf{Postcondition} \\\hline
        \bvmultbench 
        & $\set{\ \ket{s0^n0} : |s|=n\ }$ 
        & $\set{\ \ket{ss0} : |s|=n\ }$ 
        \\\hline
        \ghzmultbench
        & $\set{\ \ket i : |i|=n\ }$ 
        & $\displaystyle\set{\begin{aligned}
            &\tfrac{1}{\sqrt2}\ket{0i}+\tfrac{1}{\sqrt2}\ket{1\bar{i}},\\
            &\tfrac{1}{\sqrt2}\ket{0i}-\tfrac{1}{\sqrt2}\ket{1\bar{i}} 
          \end{aligned}\ :\ |i|=n\ }$ 
        \\\hline
        \grovermultbench
        & $\set{\ \ket{s0^n0^{n-2}0} : |s|=n\ }$ 
        & $\displaystyle
          \begin{aligned}
             \bigcup_{\substack{\text{imag}(a_h)=0,\\|a_h|^2 > 7/8}} \{\ &a_h\ket{ss0^{n-2}1} \\
            +\ &a_\ell\sum_{\substack{i\ne s}}\ket{si0^{n-2}1} : |s|=n\ \}
          \end{aligned}$
        \\\hline
        \oegroverbench 
        & $\displaystyle
          \begin{aligned}
             \bigcup_{\substack{\text{imag}(a_h)=0,\\\text{real}(a_h)>0,\\\text{imag}(a_\ell)=0,\\\text{real}(a_\ell)>0,\\7a_\ell > a_h}} \{\ &a_h\ket{ss0^{n-2}1} \\
            +\ &a_\ell\sum_{\substack{i\ne s}}\ket{si0^{n-2}1} : |s|=n\ \}
          \end{aligned}$
        & $\displaystyle
          \begin{aligned}
             \bigcup_{\substack{\text{imag}(a'_h)=0,\\\text{imag}(a'_\ell)=0,\\|a'_h| > |a_h|}} \{\ &a'_h\ket{ss0^{n-2}1} \\
            +\ &a'_\ell\sum_{\substack{i\ne s}}\ket{si0^{n-2}1} : |s|=n\ \}
          \end{aligned}$
        \\\hline
        \makecell{\mctoffolibench\\($\neg\text{flip},\ t=0$)}
        & $\set{\ \ket{i0^{n-1}0} : i\ne 1^{n}\ }$ 
        & $\set{\ \ket{i0^{n-1}0} : i\ne 1^{n}\ }$
        \\\hline
        \makecell{\mctoffolibench\\($\neg\text{flip},\ t=1$)}
        & $\set{\ \ket{i0^{n-1}1} : i\ne 1^{n}\ }$ 
        & $\set{\ \ket{i0^{n-1}1} : i\ne 1^{n}\ }$
        \\\hline
        \makecell{\mctoffolibench\\($\text{flip},\ t=0$)}
        & $\set{\ \ket{1^n0^{n-1}0}\ }$ 
        & $\set{\ \ket{1^n0^{n-1}1}\ }$
        \\\hline
        \makecell{\mctoffolibench\\($\text{flip},\ t=1$)}
        & $\set{\ \ket{1^n0^{n-1}1}\ }$ 
        & $\set{\ \ket{1^n0^{n-1}0}\ }$
        \\\hline
    \end{tabular}
    }
\end{table}

\end{document}